\documentclass[11pt]{article}
\usepackage{amssymb}
\usepackage{amsmath}
\usepackage{amsmath}
\usepackage{amsthm}
\usepackage{amsfonts}
\usepackage{graphicx}
\usepackage{bbm}
\usepackage{mathrsfs}
\usepackage{enumerate}
\usepackage{dsfont}
\usepackage{natbib}
\usepackage{multirow}
\usepackage{dsfont}
\usepackage{color}

\numberwithin{equation}{section}
\newtheorem{theorem}{Theorem}

\newtheorem{algorithm}{Algorithm}

\newtheorem{corollary}{Corollary}

\newtheorem{assumption}{Assumption}
\newtheorem{lemma}{Lemma}
\theoremstyle{definition}

\DeclareMathOperator{\E}{\text{E}}

\newcommand{\Epsilon}{\mathcal{E}}

\newcommand{\norm}[1]{\left \|#1\right \|}
\newcommand{\abs}[1]{\left\vert#1\right\vert}

\begin{document}

\title{Causal Inference by Quantile Regression Kink Designs}
\author{
Harold D. Chiang\thanks{harold.d.chiang@vanderbilt.edu. Department of Economics, Vanderbilt University.} \qquad Yuya Sasaki\thanks{yuya.sasaki@vanderbilt.edu. Department of Economics, Vanderbilt University.} \thanks{We would like to thank Patty Anderson and Bruce Meyer for kindly agreeing to our use of the CWBH data. We benefited from very useful comments by Han Hong (the editor), the associate editor, anonymous referees, Matias Cattaneo, Andrew Chesher, Antonio Galvao, Emmanuel Guerre, Blaise Melly, Jungmo Yoon, seminar participants at Academia Sinica, National Chengchi University, National University of Singapore, Otaru University of Commerce, Penn State University, Queen Mary University of London, Temple University, University of Pittsburgh, University of Surrey, University of Wisconsin-Madison and Vanderbilt University, and conference participants at AMES 2016, CeMMAP-SNU-Tokyo Conference on Advances in Microeconometrics, ESEM 2016, New York Camp Econometrics XI. All remaining errors are ours.}\\ Vanderbilt University
\bigskip\\
}
\date{
\today }
\maketitle

\begin{abstract}\setlength{\baselineskip}{6.0mm}
The quantile regression kink design (QRKD) is proposed by empirical researchers as a potential method to assess heterogeneous treatment effects under suitable research designs, but its causal interpretation remains unknown.
We propose a causal interpretation of the QRKD estimand.
Under flexible heterogeneity and endogeneity, the QRKD estimand measures a weighted average of heterogeneous marginal effects at respective conditional quantiles of outcome given a designed kink point.
In addition, we develop weak convergence results for the QRKD estimator as a local quantile process for the purpose of conducting statistical inference on heterogeneous treatment effects using the QRKD.
Applying our methods to the Continuous Wage and Benefit History Project (CWBH) data, we find significantly heterogeneous positive causal effects of unemployment insurance benefits on unemployment durations in Louisiana between 1981 and 1983.
These effects are larger for individuals with longer unemployment durations.
\end{abstract}
\newpage

\section{Introduction}

Some recent empirical research papers, including
Nielsen, S\o rensen and Taber (2010), Landais (2015), Simonsen, Skipper and Skipper (2015), Card, Lee, Pei and Weber (2016), and Dong (2016), conduct causal inference via the regression kink design (RKD).
A natural extension of the RKD with a flavor of unobserved heterogeneity is the quantile RKD (QRKD), which is the object that we explore in this paper.
Specifically, consider the quantile derivative Wald ratio of the form
\begin{equation}\label{eq:wald}
QRKD(\tau) =
\frac {\lim_{x \downarrow x_0} \frac {\partial}{\partial x} Q_{Y\mid X}(\tau\mid  x) - \lim_{x \uparrow x_0} \frac {\partial}{\partial x} Q_{Y\mid X}(\tau \mid x)}{\lim_{x \downarrow x_0} \frac {d}{d x} b(x) - \lim_{x \uparrow x_0} \frac {d}{d x} b(x)}
\end{equation}
at a design point $x_0$ of a running variable $x$, where $Q_{Y \mid X}(\tau|x) := \inf\{y: F(y|x)\ge \tau \}$ denotes the $\tau$-th conditional quantile function of $Y$ given $X=x$, and $b$ is a policy function.
Note that it is analogous to the RKD estimand of Card, Lee, Pei and Weber (2016):
\begin{equation}\label{eq:mean_wald}
RKD =
\frac {\lim_{x \downarrow x_0} \frac {\partial}{\partial x} \E[Y\mid X=x] - \lim_{x \uparrow x_0} \frac {\partial}{\partial x} \E[Y\mid X=x]}{\lim_{x \downarrow x_0} \frac {d}{d x} b(x) - \lim_{x \uparrow x_0} \frac {d}{d x} b(x)},
\end{equation}
except that the conditional expectations in the numerator are replaced by the corresponding conditional quantiles.
While the QRKD estimand (\ref{eq:wald}) is of potential interest in the empirical literature for assessment of heterogeneous treatment effects, little seems known about its econometric theories.
Specifically, Landais (2011) considers (\ref{eq:wald}), but no formal theories of identification, estimation, and inference are provided.
This paper develops causal interpretation (identification) and estimation theories for the QRKD estimand (\ref{eq:wald}).
In addition, we also present a practical guideline of robust inference by pivotal simulations, a procedure for bandwidth selection, and statistical testing of heterogeneous treatment effects based on the QRKD.

To understand our objective, consider a structural relation $y = g(b,x,\epsilon)$, where the outcome $y$ is determined by observed factors $(b,x)$ and unobserved factors $\epsilon$.
The marginal causal effect of $b$ on $y$ for individual $i$ with $(b_i,x_i,\epsilon_i)$ is quantified by $g_1(b_i, x_i, \epsilon_i)$, where $g_1$ denotes the partial derivative of $g$ with respect to the first argument.
An estimand $\theta$ has a causal interpretation at $(b,x)$ if it admits 
\begin{equation}\label{eq:causal_interpretation}
\theta = \int g_1(b,x,\epsilon) d\mu(\epsilon)
\end{equation}
for some probability measure $\mu$ whose support is contained in that of $\epsilon$.
The literature has proposed this way of causal interpretations for major statistical estimands.
Examples include the OLS slope (Yitzhaki, 1996), the two stage least squares estimand under multivalued discrete treatments (Angrist and Imbens, 1995), an IV estimand under partial equilibrium (Angrist, Graddy and Imbens, 2000), a list of most common treatment effects (Heckman and Vytlacil, 2005), and the slope of the quantile regression (Kato and Sasaki, 2017).
In a similar spirit, we argue in the present paper that the QRKD estimand (\ref{eq:wald}) can be reconciled with the causal interpretation of the form (\ref{eq:causal_interpretation}).

Making causal interpretations of the QRKD estimand (\ref{eq:wald}) in the form (\ref{eq:causal_interpretation}) is perhaps more challenging than the mean RKD estimand (\ref{eq:mean_wald}) because the differentiation operator $\frac{d}{dx}$ and the conditional quantile do not `swap.'
For the mean RKD estimand (\ref{eq:mean_wald}), the interchangeability of the differentiation operator and the expectation (integration) operator allows each term of the numerator in (\ref{eq:mean_wald}) to be additively decomposed into two parts, namely the causal effects and the endogeneity effects.
Taking the difference of two terms in the numerator then cancels out the endogeneity effects, leaving only the causal effects.
This trick allows the mean RKD estimand (\ref{eq:mean_wald}) to have causal interpretations in the presence of endogeneity.
Due to the lack of such interchangeability for the case of quatiles, this trick is not straightforwardly inherited by the quantile counterpart (\ref{eq:wald}).
Having said this, we show in Section \ref{sec:causal_interpretation} that a similar decomposition is possible for the QRKD estimand (\ref{eq:wald}), and therefore argue that its causal interpretations are possible even under the lack of monotonicity.
Specifically, we show that the QRKD estimand corresponds to the quantile marginal effect under monotonicity and to a weighted average of marginal effects under non-monotonicity.

For estimation of the causal effects, we propose a sample-counterpart estimator for the QRKD estimand (\ref{eq:wald}) in Sections \ref{sec:estimation}.
To derive its asymptotic properties, we take advantage of the existing literature on uniform Bahadur representations for quantile-type loss functions, including Kong, Linton and Xia (2010), Guerre and Sabbah (2012), Sabbah (2014), and Qu and Yoon (2015a).
Qu and Yoon (2015b) apply the results of Qu and Yoon (2015a) to develop methods of statistical inference with quantile regression discontinuity designs (QRDD), which are closely related to our QRKD framework.
We take a similar approach with suitable modifications to derive asymptotic properties of our QRKD estimator.
Weak convergence results for the estimator as quantile processes are derived.
Applying the weak convergence results, we propose procedures for testing treatment significance and treatment heterogeneity following Koenker and Xiao (2002), Chernozhukov and Fern\'andez-Val (2005) and Qu and Yoon (2015b).
Simulation studies presented in Section \ref{sec:mc} support the theoretical properties.

Literature:
The method studied in this paper falls in the broad framework of design-based causal inference, including RDD and RKD.
There is an extensive body of literature on RDD by now -- see a historical review by Cook (2008) and surveys in the special issue of \textit{Journal of Econometrics} edited by Imbens and Lemieux (2008), Imbens and Wooldridge (2009; Sec. 6.4), Lee and Lemieux (2010), and Volume 38 of \textit{Advances in Econometrics} edited by Cattaneo and Escanciano (2016), as well as the references cited therein.
The first extension to quantile treatment effects in the RDD framework was made by Frandsen, Fr\"olich and Melly (2012).
More recently, Qu and Yoon (2015b) develop uniform inference methods with QRDD that empirical researchers can use to test a variety of important empirical questions on heterogeneous treatment effects.
While the RDD has a rich set of empirical and theoretical results including the quantile extensions, the RKD method which developed more recently does not have a quantile counterpart in the literature yet, despite potential demands for it by empirical researchers (e.g., Landais, 2011).
Our paper can be seen as a quantile extension to Card, Lee, Pei and Weber (2016) and a RKD counterpart of Qu and Yoon (2015b).

\section{Causal Interpretation of the QRKD Estimand}\label{sec:causal_interpretation} 

In this section, we develop some causal interpretations of the QRKD estimand (\ref{eq:wald}).
For the purpose of illustration, we first present a simple case with rank invariance in Section \ref{sec:identification_rank_invariance}.
It is followed by a formal argument for general cases in Section \ref{sec:identification}.

\subsection{Illustration: Causal Interpretation under Rank Invariance}\label{sec:identification_rank_invariance}

The causal relation of interest is represented by the structural equation
$$y=g(b,x,\epsilon).$$
The outcome $y$ is determined through the structural function $g$ by two observed factors, $b \in \mathbb{R}$ and $x \in \mathbb{R}$, and a scalar unobserved factor, $\epsilon \in \mathbb{R}$.
We assume that $g$ is monotone increasing in $\epsilon$, effectively imposing the rank invariance; causal interpretations in a more general setup with non-monotone $g$ and/or multivariate $\epsilon$ is established in Section \ref{sec:identification}.
The factor $b$ is a treatment input, and is in turn determined by the running variable $x$ through the structural equation
$$b = b(x)$$
for a known policy function $b$.
We say that $b$ has a kink at $x_0$ if $b'(x_0^+) := \lim_{x \to x_0^+}\frac{d b(x)}{d x} \ne \lim_{x \to x_0^-}\frac{d b(x)}{d x} =: b'(x_0^-)$ is true, where $x \rightarrow x_0^+$ and $x \rightarrow x_0^-$ mean $x \downarrow x_0$ and $x \uparrow x_0$, respectively.
Throughout this paper, we assume that the location, $x_0$, of the kink is known from a policy-based research design, as is the case with Card, Lee, Pei and Weber (2016).

\begin{assumption}\label{a:policy treatment}
$b'(x_0^+)\ne b'(x_0^-)$ holds, and $b$ is continuous on $\mathbb{R}$ and differentiable on $\mathbb{R} \setminus \{x_0\}$.
\end{assumption}

The structural partial effects are $g_1(b,x,\epsilon) := \frac{\partial}{\partial b}g(b,x,\epsilon)$, $g_2(b,x,\epsilon) := \frac{\partial}{\partial x}g(b,x,\epsilon) $ and $g_3(b,x,\epsilon) := \frac{\partial}{\partial \epsilon}g(b,x,\epsilon) $.
In particular, a researcher is interested in $g_1$ which measures heterogeneous partial effects of the treatment intensity $b$ on an outcome $y$.
While the structural partial effect $g_1$ is of interest, it is not clear if the QRKD estimand (\ref{eq:wald}) provides any information about $g_1$.
In this section, we argue that (\ref{eq:wald}) does have a causal interpretation in the sense that it measures the structural causal effect $g_1(b(x_0),x_0,\epsilon)$ at the $\tau$-th conditional quantile of $\varepsilon$ given $X=x_0$.

Under regularity conditions (to be discussed in Section \ref{sec:identification} in detail), some calculations yield the decomposition
\begin{equation}\label{eq:intuitive_decomposition}
\frac{\partial}{\partial x}Q_{Y \mid X}(\tau \mid x)
=
g_1(b(x),x,\epsilon) \cdot b^\prime(x) + g_2(b(x),x,\epsilon)
- \frac{\int_{-\infty}^{\epsilon} \frac{\partial}{\partial x} f_{\varepsilon \mid X}(e \mid x) de}{f_{\varepsilon \mid X}(\epsilon \mid x)} \cdot g_3(b(x),x,\epsilon),
\end{equation}
where $\tau = F_{\varepsilon \mid X}(\epsilon \mid x)$.
The first term on the right-hand side is the partial effect of the running variable $x$ on the outcome $y$ through the policy function $b$.
The second term is the direct partial effect of the running variable on the outcome $y$.
The third term measures the effect of endogeneity in the running variable $x$.
We can see that this third term is zero under exogeneity, $\frac{\partial}{\partial x} f_{\varepsilon \mid X} = 0$.
In order to get the causal effect $g_1(b(x),x,\epsilon)$ of interest through the QRKD estimand (\ref{eq:wald}), therefore, we want to remove the last two terms in (\ref{eq:intuitive_decomposition}).

Suppose that the designed kink condition of Assumption \ref{a:policy treatment} is true, but all the other functions, $g_1$, $g_2$, $g_3$, $1/f_{\varepsilon \mid X}$ and $\frac{\partial}{\partial x} f_{\varepsilon \mid X}$, in the right-hand side of (\ref{eq:intuitive_decomposition}) are continuous in $(b,x)$ at $(b(x_0),x_0)$.
Then, (\ref{eq:intuitive_decomposition}) yields
\begin{equation}\label{eq:id:monotone}
\frac {\frac {\partial}{\partial x} Q_{Y\mid X}(\tau\mid  x_0^+) - \frac {\partial}{\partial x} Q_{Y\mid X}(\tau \mid x_0^-)}{b^\prime(x_0^+) - b^\prime(x_0^-)}
= g_1(b(x_0),x_0,\epsilon),
\end{equation}
showing that the QRKD estimand (\ref{eq:wald}) measures the structural causal effect $g_1(b(x_0),x_0,\epsilon)$ of $b$ on $y$ for the subpopulation of individuals at the $\tau$-th conditional quantile of $\varepsilon$ given $X=x_0$.
This section provides only an informal argument for ease of exposition, but Section \ref{sec:identification} provides a formal mathematical argument under a general setup without the rank invariance assumption.

\subsection{General Result: Causal Interpretation without Rank Invariance}\label{sec:identification}

In this section, we continue to use the basic settings from Section \ref{sec:identification_rank_invariance} except that the unobserved factors $\epsilon$ are now allowed to be $M$-dimensional, as opposed to be a scalar, and that $g$ is now allowed to be non-monotone with respect to any coordinate of $\epsilon$.
As such, we can consider general structural functions $g$ without the rank invariance.
In this case, there can exist multiple values of $\epsilon$ corresponding to a single conditional quantile $\tau$ of $Y$ given $X=x_0$, and therefore the simple identifying equality (\ref{eq:id:monotone}) for the case of rank invariance cannot be established in general.
Furthermore, $QRKD(\tau)$ even fails to equal the average of the structural derivatives $g_1(b(x_0),x_0,\epsilon)$ for those $\epsilon$ that coincide with the $\tau$-th conditional quantile of $Y$ given $X=x_0$.
Nonetheless, we argue that $QRKD(\tau)$ represents a weighted average of the structural derivatives $g_1(b(x_0),x_0,\epsilon)$ for those $\epsilon$ that coincide with the $\tau$-th conditional quantile of $Y$ given $X=x_0$.

Define the lower contour set of $\epsilon$ evaluated by $g(b(x),x,\cdot)$ below a given level of $y$ as follows: 
$$V(y,x)=\{\epsilon \in \mathbb{R}^M |g(b(x),x,\epsilon)\le y \}.$$
Its boundary is denoted by $\partial V(y,x)$.
Furthermore, the velocities of the boundary $\partial V(y,x)$ at $\epsilon$ with respect to a change in $y$ and a change in $x$ are denoted by $\partial \upsilon(y,x;\epsilon)/\partial y$ and $\partial \upsilon(y,x;\epsilon)/\partial x$, respectively.
For a short hand notation, we write $h(x,\epsilon)=g(b(x),x,\epsilon)$ and $h_x (x,\epsilon)=\frac{\partial h(x,\epsilon)}{\partial x}$.
Under regularity conditions to be stated below, the implicit function theorem allows the velocities defined above to be explicitly written as $\partial \upsilon(y,x;\epsilon)/\partial y = 1/||\nabla_\epsilon h(x,\epsilon)||$ and $\partial \upsilon(y,x;\epsilon)/\partial x = -h_x(x,\epsilon)/||\nabla_\epsilon h(x,\epsilon)||$ for all $\epsilon \in V(y,x)$.
Let $\Sigma$ denote an $(M-1)$-dimensional rectangle, and we parameterize the manifold $\partial V(y,x)$ by $\Pi_{y,x}: \Sigma \rightarrow \partial V(y,x)$ for all $(y,x)$.
We refer to Padula (2011) for further details of these objects and notations.
Let $m^M$ and $H^{M-1}$ denote the Lebesgue measure on $\mathbb{R}^M$ and the Hausdorff measure\footnote{The Hausdorff measure is defined as follows. Define a function $H_p^{M-1}: 2^{\mathbb{R}^M} \rightarrow \mathbb{R}$ by $H_p^{M-1}(S) = \sup_{\delta>0} \inf \left\{ \sum_{i=1}^\infty (\text{diam} S_i)^{M-1} \ \left\vert \ \cup_{i=1}^\infty S_i \supset S, \ \text{diam} S_i < \delta \right.\right\}$.
We then define a restriction $H^{M-1}: \mathcal{B}(V) \rightarrow \mathbb{R}$ of $H_p^{M-1}$ to the Borel sigma algebra $\mathcal{B}(V)$ of a metric space $V \subset \mathbb{R}^M$ is a measure, and we call it the $(M-1)$-dimensional Hausdorff measure. Intuitively, $H^{M-1}$ measures a scaled area of Borel subsets of the $(M-1)$-dimensional manifold $V \subset \mathbb{R}^M$.} on $\partial V(y,x)$, respectively.
Letting $\mathcal{X} = supp(X)$, we make the following assumptions.

\begin{assumption}\label{a:regularity}
(i) $h( \cdot , \epsilon)$ is continuously differentiable on $\mathcal{X} \setminus \{ x_0 \}$ for all $\epsilon \in \Epsilon$ and $h(x,\cdot)$ is continuously differentiable for all $x \in \mathcal{X}$. 
(ii) $\norm{\nabla_\epsilon h(x, \ \cdot)} \neq 0$ on $\partial V(y,x)$ for all $(x,y) \in \mathcal{X} \times \mathcal{Y}$.
(iii) The conditional distribution of $\varepsilon$ given $X$ is absolutely continuous with respect to $m^M$, $f_{\varepsilon \mid X}$ is continuously differentiable, and $\mathcal{X} \ni x \mapsto f_{\varepsilon \mid X}( \cdot | x ) \in L^1(\mathbb{R}^M,m^M)$ is continuous.\footnote{That is, for all $\delta_1 > 0$ there exists $\delta_2 > 0$ such that $\abs{x'-x} < \delta_2$ implies $\int \abs{f_{\varepsilon \mid X}( \epsilon | x' ) - f_{\varepsilon \mid X}( \epsilon | x )} dm^M(\epsilon) < \delta_1$.}
(iv) $\int_{\partial V(y,x)} f_{\varepsilon \mid X}(\epsilon \mid x) dH^{M-1}(\epsilon) > 0$ for all $(x,y) \in \mathcal{X} \times \mathcal{Y}$.
\end{assumption}

\begin{assumption}\label{a:boundary}
(i) For $M = 1$: $\partial V(y,x)$ is a finite set, and $h(x, \cdot)$ is locally invertible with a continuously differentiable local inverse function in a neighborhood of each point in $\partial V(y,x)$. 
(ii) For $M > 1$: 
$\Sigma \times \mathcal{X} \ni (s,x) \mapsto \Pi_{y,x}(s) \in \mathbb{R}^M$ is continuous for all $y \in \mathcal{Y}$, and $\Sigma \times \mathcal{Y} \ni (s,y) \mapsto \Pi_{y,x}(s) \in \mathbb{R}^M$ is continuous for all $x \in \mathcal{X}$.
$\mathcal{X} \ni x \mapsto \partial \upsilon(y,x ;\Pi_{y,x}(\cdot))/\partial x \in L^1(\Sigma,m^{M-1})$ is continuous for all $y \in \mathcal{Y}$.\footnote{That is, $\forall \delta_1 > 0$ $\exists\delta_2 > 0$ such that $\abs{x'-x} < \delta_2$ implies $\int \abs{\partial \upsilon(y,x ;\Pi_{y,x'}(s))/\partial x - \partial \upsilon(y,x ;\Pi_{y,x}(s))/\partial x }ds < \delta_1$.}
$\mathcal{Y} \ni y \mapsto \partial \upsilon(y,x ;\Pi_{y,x}(\cdot))/\partial y \in L^1(\Sigma,m^{M-1})$ is continuous for all $x \in \mathcal{X}$.\footnote{That is, $\forall \delta_1 > 0$ $\exists\delta_2 > 0$ such that $\abs{y'-y} < \delta_2$ implies $\int \abs{\partial \upsilon(y,x ;\Pi_{y',x}(s))/\partial y - \partial \upsilon(y,x ;\Pi_{y,x}(s))/\partial y }ds < \delta_1$.}
\end{assumption}

\begin{assumption}\label{a:pq}
Let $\gamma(x,\epsilon) := \norm{\nabla_\epsilon h(x,\epsilon)}^{-1}$.
There exist $p \geqslant 1$ and $q \geqslant 1$ satisfying $p^{-1}+q^{-1} = 1$ such that $\norm{\gamma(x, \ \cdot \ )}_{L^p(\partial V(y,x), H^{M-1})} < \infty$ and $\norm{f_\varepsilon}_{L^q(\partial V(y,x), H^{M-1})} < \infty$ hold for all $(x,y) \in \mathcal{X} \times \mathcal{Y}$.
\end{assumption}

\begin{assumption}\label{a:DCT}
(i)
There exists $w_{y,x} \in L^1 (\partial V (y,x), H^{M-1})$ such that
 $|\gamma(x,\epsilon) h_x(x,\epsilon) f_{\varepsilon | X} (\epsilon |x) |\le w_{y,x}(\epsilon)$ and
$|\gamma(x,\epsilon) f_{\varepsilon | X} (\epsilon |x) |\le w_{y,x}(\epsilon)$ for all $\epsilon \in \partial V (y,x)$ for all $(y,x) \in \mathcal{Y} \times \mathcal{X}$.
(ii)
There exists $w_{y,x} \in L^1 (V (y,x), m^{M-1})$ such that
 $|\partial f_{\varepsilon | X} (\epsilon |x) / \partial x |\le w_{y,x}(\epsilon)$ for all $\epsilon \in V (y,x)$ for all $(y,x) \in \mathcal{Y} \times \mathcal{X}$.
\end{assumption}


Assumptions \ref{a:regularity}, \ref{a:boundary} and \ref{a:pq} are used to derive a structural decomposition of the quantile partial derivative -- see Sasaki (2015) for detailed discussions of these assumptions.
Assumption \ref{a:boundary} branches into two cases, depending on (i) $M=1$ or (ii) $M>1$.
We note that case (i) accommodates a non-monotone structure $g$ in a scalar unobservable $\varepsilon$, whereas case (ii) concerns about non-monotonicity due to multi-dimensional unobservables $\varepsilon$.
These two cases are stated separate because the restriction in case (ii) among others entails that $\partial V(y,x)$ is a connected set, which is too strong for case (i) with non-monotonicity.
In Assumptions \ref{a:regularity} (iv) and \ref{a:pq}, statements concern about integration of $f_{\varepsilon | X=x}$ on $\partial V(y,x)$.
This manifold $\partial V(y,x)$ has a Lebesgue measure zero, i.e., $m^M(\partial V(y,x))=0$.
On the other hand, the Hausdorff measure evaluates this Lebesgue null set positively, i.e.,  $H^{M-1}(\partial V(y,x)) > 0$.
Hence these assumptions are nontrivial statements.

The regularity conditions in Assumption \ref{a:DCT} facilitate the dominated convergence theorem to make a structural sense of the QRKD estimand (\ref{eq:wald}). Specifically, by the dominated convergence theorem, Assumption 5 (i) and (ii) together with Assumption 2 (iv) and 4 are sufficient for the existence of the reduced-form expressions $lim_{x\to x_0^+}\frac{\partial}{\partial x} Q_{Y|X}(\tau|x)$ and $lim_{x\to x_0^-}\frac{\partial}{\partial x} Q_{Y|X}(\tau|x)$.
With $\mathcal{B}(y,x)$ denoting the collection of Borel subsets of $\partial V(y,x)$, we define the function $\mu^{M-1}_{y,x}: \mathcal{B}(y,x) \rightarrow \mathbb{R}$ by
\begin{displaymath}
\mu^{M-1}_{y,x}(S) :=
\frac{\int_{s} \frac{1}{\norm{\nabla_\epsilon h(x,\epsilon)}} f_{\varepsilon |X} (\epsilon|x) dH^{M-1}(\epsilon)}
{\int_{\partial V(y,x)} \frac{1}{\norm{\nabla_\epsilon h(x,\epsilon)}} f_{\varepsilon|X} (\epsilon|x) dH^{M-1}(\epsilon)}
\qquad
\text{for all } S \in \mathcal{B}(y,x).
\end{displaymath}
Because the zero-dimensional Hausdorff measure $H^0$ is a counting measure, the case of $M=1$ yields
\begin{equation}\label{eq:mu0}
\mu^{0}_{y,x}(\{\epsilon\}) :=
\frac{\frac{1}{\norm{\nabla_\epsilon h(x,\epsilon)}} f_{\varepsilon |X} (\epsilon|x)}
{\sum_{\epsilon \in \partial V(y,x)} \frac{1}{\norm{\nabla_\epsilon h(x,\epsilon)}} f_{\varepsilon|X} (\epsilon|x)}
\qquad
\text{for all } \epsilon \in \partial V(y,x).
\end{equation}
The next theorem claims that this is a probability measure and gives weights with respect to which the QRKD estimand (\ref{eq:wald}) measures the average structural causal effect of the treatment intensity $b$ on an outcome $y$ for those individuals at the $\tau$-th conditional quantile of $Y$ given $X=x_0$.

\begin{theorem}\label{theorem:identification}
Suppose that Assumptions \ref{a:policy treatment}, \ref{a:regularity}, \ref{a:boundary}, \ref{a:pq} and \ref{a:DCT}  hold.
Let $\tau \in (0,1)$ and $y = Q_{Y | X}(\tau | x_0)$.
Then, $\mu^{M-1}_{y,x_0}$ is a probability measure on $\partial V(y,x_0)$, and
\begin{equation}\label{eq:id}
QRKD(\tau)
\ = \
\int_{\partial V(y,x_0)} g_1(b(x_0),x_0,\epsilon) d\mu^{M-1}_{y,x_0}(\epsilon)
\ = \
E_{\mu^{M-1}_{y,x_0}} \left[ g_1(b(x_0),x_0,\varepsilon)\right].
\end{equation}
\end{theorem}

\begin{proof}
For the first part of the proof, we branch into two cases: (i) $M=1$ and (ii) $M>1$.
\\
(i) For $M = 1$: 
That $\mu^{M-1}_{y,x}$ is a probability measure on $\partial V(y,x)$ follows from (\ref{eq:mu0}) under Assumption \ref{a:pq}.
By Leibniz integral rule and the implicit function theorem under Assumptions \ref{a:regularity}, \ref{a:boundary} (i) and \ref{a:pq}, the QPD $\frac{\partial}{\partial x} Q_{Y\mid X}(\tau \mid x)$ exists and
\begin{align*}
\frac{\partial}{\partial x} Q_{Y \mid X}(\tau \mid x) &=
\frac{\sum_{\epsilon \in \partial V(y,x)} \frac{h_x(x,\epsilon)}{\abs{h_\epsilon(x,\epsilon)}} f_{\varepsilon \mid X} (\epsilon \mid x) - \int_{V(y,x)} \frac{\partial}{\partial x} f_{\varepsilon \mid X}(\epsilon \mid x) d\epsilon}
{\sum_{\epsilon \in \partial V(y,x)} \frac{1}{\abs{h_\epsilon(x,\epsilon)}} f_{\varepsilon \mid X} (\epsilon \mid x)}\\
&=\E_{\mu^{0}_{y,x}} [h_x(x,\varepsilon)]-A(y,x),
\end{align*}
where $A$ is defined by
\begin{equation*}
A(y,x) :=
\frac{\int_{V(y,x)} \frac{\partial}{\partial x} f_{\varepsilon \mid X}(\epsilon \mid x) d\epsilon}
{\sum_{\epsilon \in \partial V(y,x)} \frac{1}{\abs{h_\epsilon(x,\epsilon)}} f_{\varepsilon \mid X} (\epsilon \mid x)}
\end{equation*}
(ii) For $M > 1$: 
That $\mu^{M-1}_{y,x}$ is a probability measure on $\partial V(y,x)$ follows from Lemma 2 of Sasaki (2015) under Assumption \ref{a:pq}.
Next, by Lemma 1 of Sasaki (2015) under Assumptions \ref{a:regularity}, \ref{a:boundary} (ii) and \ref{a:pq}, the QPD $\frac{\partial}{\partial x} Q_{Y\mid X}(\tau \mid x)$ exists and
\begin{align*}
\frac{\partial}{\partial x} Q_{Y \mid X}(\tau \mid x) &=
\frac{\int_{\partial V(y,x)} \frac{h_x(x,\epsilon)}{\norm{\nabla_\epsilon h(x,\epsilon)}} \frac{f_{\varepsilon \mid X} (\epsilon \mid x) \cdot M \pi^{(M-1)/2}}{2^{M-1} \Gamma(\frac{M+1}{2})} dH^{M-1}(\epsilon) - \int_{V(y,x)} \frac{\partial}{\partial x} f_{\varepsilon \mid X}(\epsilon \mid x) dm^M(\epsilon)}
{\int_{\partial V(y,x)} \frac{1}{\norm{\nabla_\epsilon h(x,\epsilon)}} \frac{f_{\varepsilon \mid X} (\epsilon \mid x) \cdot M \pi^{(M-1)/2}}{2^{M-1} \Gamma(\frac{M+1}{2})} dH^{M-1}(\epsilon)}\\
&=\E_{\mu^{M-1}_{y,x}} [h_x(x,\varepsilon)]-A(y,x),
\end{align*}
where $\Gamma$ is the Gamma function and $A$ is defined by
\begin{equation*}
A(y,x) :=
\frac{ \int_{V(y,x)} \frac{\partial}{\partial x} f_{\varepsilon \mid X}(\epsilon \mid x) dm^M(\epsilon)}
{\int_{\partial V(y,x)} \frac{1}{\norm{\nabla_\epsilon h(x,\epsilon)}} \frac{f_{\varepsilon \mid X} (\epsilon \mid x) \cdot M \pi^{(M-1)/2}}{2^{M-1} \Gamma(\frac{M+1}{2})} dH^{M-1}(\epsilon)}
\end{equation*}

From this point on, we treat both cases (i) $M=1$ and (ii) $M>1$ together.
Note that $g_2 = \frac{\partial g}{\partial x}$ is continuous in $x$ by Assumption \ref{a:regularity} (i).
Also, $\mu^{M-1}_{y,x}(\epsilon)$ is continuous in $x$ for each fixed $y$ according to parts (i), (ii) and (iii) of Assumption \ref{a:regularity}.
Furthermore, Assumption \ref{a:regularity} (i), (ii), (iii) and (iv) imply that $A(y,x)$ is well-defined and is continuous in $x$ for all $y\in \mathcal{Y}$.
Therefore, applying the dominated convergence theorem under Assumptions \ref{a:regularity} (iv), \ref{a:pq} and \ref{a:DCT} yields
\begin{eqnarray*}
\lim_{x\to x_0^+} \frac {\partial}{\partial x}Q_{Y\mid X}(\tau \mid  x)
&=& \lim_{x\to x_0^+} \int_{\partial V(y,x)}  \{ h_x (x,\epsilon)\} d\mu^{M-1}_{y,x}(\epsilon)-\lim_{x\to x_0^+}A(y,x) \\
&=& \int_{\partial V(y,x_0)} \lim_{x\to x_0^+} \frac{\partial}{\partial x} \{ g(b(x),x,\epsilon)\} d\mu^{M-1}_{y,x_0}(\epsilon)-A(y,x_0) \\
&=& \int_{\partial V(y,x_0)} \lim_{x\to x_0^+} \{ g_1 (b(x),x,\epsilon)b'(x) + g_2 (b(x),x,\epsilon)  \} d\mu^{M-1}_{y,x_0}(\epsilon)- A(y,x_0) \\
&=&  \int \{ g_1 (b( x_0), x_0,\epsilon)b'( x_0^+) + g_2 (b( x_0), x_0,\epsilon)  \} d\mu^{M-1}_{y, x_0}(\epsilon)-A(y, x_0)
\end{eqnarray*}
Similarly, taking the limit from the left, we have
\begin{align*}
\lim_{x\to x_0^-} \frac {\partial}{\partial x}Q_{Y\mid X}(\tau \mid  x) =& \int_{\partial V(y,x_0)} \{ g_1 (b( x_0), x_0,\epsilon)b'( x_0^-) + g_2 (b( x_0), x_0,\epsilon)  \} d\mu^{M-1}_{y, x_0}(\epsilon)-A(y, x_0).
\end{align*}
Taking the difference of the right and left limits eliminates $\int_{\partial V(y,x_0)} g_2 (b( x_0), x_0,\epsilon) d\mu^{M-1}_{y, x_0}(\epsilon)-A(y, x_0)$, and thus produces
\begin{align*}
\lim_{x\to x_0^+} \frac {\partial}{\partial x} Q_{Y\mid X}(\tau\mid  x)-\lim_{x\to x_0^-} \frac {\partial}{\partial x} Q_{Y\mid X}(\tau \mid  x)
&= [b'(x_0^+)- b'(x_0^-)] E_{\mu^{M-1}_{y,x_0}} \left[ g_1(b(x_0),x_0,\varepsilon)\right].
\end{align*}
Finally, note that Assumption \ref{a:policy treatment} has $b'(x_0^+)- b'(x_0^-)\ne 0 $, and hence we can divide both sides of the above equality by $b'(x_0^+)- b'(x_0^-) $.
This gives the desired result.
\end{proof}

As is often the case in the treatment literature (e.g., Angrist and Imbens, 1995), this theorem shows a causal interpretation in terms of a weighted average.
Specifically, (\ref{eq:id}) shows that the QRKD estimand (\ref{eq:wald}) measures a weighted average of the heterogeneous causal effects $g_1(b(x_0),x_0,\epsilon)$ displayed on the right-hand side of (\ref{eq:id}).
Since the weights are positive on the support of the conditional distribution of $\varepsilon$ given $X=x_0$, the QRKD estimand is a strict convex combination of the ceteris paribus causal effects of $b$ on $y$ for those individuals at the $\tau$-th conditional quantile of $Y$ given $X=x_0$.

The weights given in the definition of $\mu^{M-1}_{y,x_0}$ are proportional to
$
{f_{\varepsilon |X} (\epsilon|x_0)} / {\norm{\nabla_\epsilon h(x_0,\epsilon)}}.
$
Since $f_{\varepsilon |X} (\epsilon|x_0)$ is the conditional density of the unobservables $\varepsilon$ given $X=x_0$, the discrepancy between the weighted and unweighted averages is imputed to the denominator, $\norm{\nabla_\epsilon h(x_0,\epsilon)}$.
For example, larger weights are assigned to those locations of $\epsilon \in \partial V(y,x_0)$ at which $\norm{\nabla_\epsilon h(x_0,\epsilon)}$ is smaller.
In other words, the QRKD emphasizes those locations of $\epsilon \in \partial V(y,x_0)$ at which the effects of unobservables $\varepsilon$ on the structure $g$ are smaller in magnitude.
On the other hand, the QRKD de-emphasizes those locations of $\epsilon \in \partial V(y,x_0)$ at which the effects of unobservables $\varepsilon$ on the structure $g$ are larger in magnitude.

One may worry about the obscurity of the causal interpretations under the `weighted' averages.
Note that the weighted average becomes an unweighted average when ${\norm{\nabla_\epsilon h(x_0,\epsilon)}}$ is constant in $\epsilon$.
There are some cases where the weight is constant.
As an example which is often relevant to empirical practices, the polynomial random coefficient models of the form
\begin{equation}\label{eq:random_coefficient}
g(b,x,\epsilon) = \epsilon_{00} + \sum_{\nu=1}^{p_b} \epsilon_{\nu 0} b^\nu + \sum_{\nu=1}^{p_x} \epsilon_{0 \nu} x^\nu + \sum_{\nu_b=1}^{p_b} \sum_{\nu_x=1}^{p_x} \epsilon_{\nu_b \nu_x} b^{\nu_b} x^{\nu_x}
\end{equation}
satisfies that ${\norm{\nabla_\epsilon h(x_0,\epsilon)}}$ is constant in $\epsilon = (\epsilon_{00},\epsilon_{10},\ldots,\epsilon_{p_b 0},\epsilon_{01},\ldots,\epsilon_{0 p_x},\epsilon_{11},\ldots,\epsilon_{p_b p_x})$.
Therefore, we obtain the following unweighted average causal interpretation for the QRKD estimand under this model.

\begin{corollary}
Suppose that the assumptions for Theorem \ref{theorem:identification} hold with (\ref{eq:random_coefficient}).
Let $\tau \in (0,1)$ and $y = Q_{Y | X}(\tau | x_0)$.
Then,
$$
QRKD(\tau)
\ = \
\int_{\partial V(y,x_0)} g_1(b(x_0),x_0,\epsilon) d\mu^{M-1}_{y,x_0}(\epsilon)
\ = \
E_{\mu^{M-1}_{y,x_0}} \left[ g_1(b(x_0),x_0,\varepsilon)\right].
$$
where
$$
\mu^{M-1}_{y,x}(S) :=
\frac{\int_{s} f_{\varepsilon |X} (\epsilon|x) dH^{M-1}(\epsilon)}
{\int_{\partial V(y,x)} f_{\varepsilon|X} (\epsilon|x) dH^{M-1}(\epsilon)}
\qquad
\text{for all } S \in \mathcal{B}(y,x).
$$
\end{corollary}

When the unobservable $\varepsilon$ is a scalar random variable (i.e., $M=1$), the Hausdorff measure $H^{M-1}$ becomes a counting measure $H^0$ on the zero-dimensional manifold $\partial V(y,x_0) \subset \mathbb{R}$.
In that case, (\ref{eq:id}) may be rewritten as
\begin{equation}\label{eq:weighted_sum}
QRKD(\tau)
\ = \
\sum_{\epsilon \in \partial V(y,x_0)} g_1(b(x_0),x_0,\epsilon) \cdot \mu^{M-1}_{y,x_0}(\{\epsilon\})
\ = \
E_{\mu^{M-1}_{y,x_0}} \left[ g_1(b(x_0),x_0,\varepsilon)\right].
\end{equation}
In particular, the case where $\partial V(y,x_0)$ is a singleton allows for the following straightforward causal interpretation for the QRKD estimand.

\begin{corollary}\label{cor:local_monotone}
Suppose that the assumptions for Theorem \ref{theorem:identification} hold with (\ref{eq:random_coefficient}).
Let $\tau \in (0,1)$ and $y = Q_{Y | X}(\tau | x_0)$.
If $\varepsilon$ is a sclar radom variable (i.e., $M=1$) and $\partial V(y,x_0)$ is a singleton, then
$$
QRKD(\tau)
\ = \
g_1(b(x_0),x_0,\epsilon(y,x_0)),
$$
where $\epsilon(y,x_0)$ is the sole element of $\partial V(y,x_0)$.
\end{corollary}

Note that this corollary is a generalization of (\ref{eq:id:monotone}), and admits the straightforward causal interpretation $QRKD(\tau) = g_1(b(x_0),x_0,\epsilon(y,x_0))$ without requiring the `global' monotonicity of $g$ in $\epsilon$.
To see the point in case, consider the structural function given by 
$$
g(b,x,\epsilon) = - 9b\epsilon + \frac{1}{3} b\epsilon^3 - 9x\epsilon + \frac{1}{3} x\epsilon^3.
$$
If $b(x_0)+x_0 \neq 0$, then this structure is not globally monotone in $\epsilon$ at $x=x_0$.
However, $\partial V(y,x_0)$ is a singleton (i.e., $g(b(x_0),x_0,\cdot)$ is locally monotone) for each value of $y \not\in [-18,18]$, and hence the causal interpretation $QRKD(\tau) = g_1(b(x_0),x_0,\epsilon(y,x_0))$ of Corollary \ref{cor:local_monotone} applies.
On the other hand, for each value of $y \in [-18,18]$, we can interpret the QRKD at most in terms of the weighted sum of the form (\ref{eq:weighted_sum}).

In either of these cases, heterogeneity in values of the QRKD estimand across quantiles $\tau$ can be used as evidence for heterogeneity in treatment effects.
Therefore, we can still conduct statistical inference for heterogeneous treatment effects based on the weak convergence results presented below in Section \ref{sec:estimation}.

\section{Estimation and Inference}\label{sec:estimation}

\subsection{The Estimator and Its Asymptotic Distribution}

We propose to estimate the QRKD estimand (\ref{eq:wald}) by its sample counterpart
\begin{equation}\label{eq:wald_estimator}
\widehat{QRKD}(\tau) \ = \
\frac {\hat{\beta}^+_1(\tau) - \hat{\beta}^-_1(\tau)}{ b^\prime (x^+_0) -  b^\prime(x^-_0)},
\end{equation}
where the two terms in the numerator are given by the $p$-th order local polynomial quantile smoothers
\begin{align*}
\hat \beta^+_1(\tau)=&\iota'_2\underset{(\alpha, \beta^+_1, \beta^-_1,...,\beta^+_p, \beta^-_p)\in \mathds{R}^{2p+1}}{\text{argmin}} \sum_{i=1}^{n}K\Big(\frac{x_i-x_0}{h_{n,\tau}}\Big)\rho_\tau \Big( y_i- \alpha - \sum_{v=1}^{p}(\beta^+_v d^+_i + \beta^-_v d^-_i)\frac{(x_i-x_0)^v}{v!} \Big)\\
\hat \beta^-_1(\tau)=&\iota'_3\underset{(\alpha, \beta^+_1, \beta^-_1,...,\beta^+_p, \beta^-_p)\in \mathds{R}^{2p+1}}{\text{argmin}} \sum_{i=1}^{n}K\Big(\frac{x_i-x_0}{h_{n,\tau}}\Big)\rho_\tau \Big( y_i- \alpha - \sum_{v=1}^{p}(\beta^+_v d^+_i + \beta^-_v d^-_i)\frac{(x_i-x_0)^v}{v!} \Big)
\end{align*}
for $\tau \in T$, where $T \subset (0,1)$ is a closed interval, $K$ is a kernel function, $\rho_\tau(u)=u(\tau-\mathds{1}\{u <0\})$,
 $d^+_i = \mathds{1}\{x_i> x_0\}$,
$d^-_i = \mathds{1}\{x_i< x_0\}$, and
$\iota_2 = [0,1,0,0,...,0]' $, $\iota_3 = [0,0,1,0,...,0]'\in \mathds{R}^{2p+1}$ for a fixed integer $p\ge 1$ of polynomial order. Notice that we are imposing the constraint that conditional quantile function $Q_{Y|X}(\tau|x)$ is continuous at $x_0$ like the estimator of Landais (2011). 
A researcher observing a sample $\{y_i,x_i\}_{i=1}^n$ of $n$ observations can compute (\ref{eq:wald_estimator}) to estimate (\ref{eq:wald}).

Our motivation to include the higher order terms in the local polynomial estimation is to implement a one-step bias correction for a local linear estimation that can accommodate optimal bandwidths -- see Remark 7 in Calonico, Cattaneo and Titiunik (2014) and Remark S.A.7 in the supplementary appendix of Calonico, Cattaneo and Titiunik (2014). 
That is, this estimator can be considered as the one-step bias corrected version of the local linear quantile smoother ($p=1$):
\begin{align*}
\underset{(\alpha, \beta^+_1, \beta^-_1)\in \mathds{R}^{3}}{\text{argmin}} \sum_{i=1}^{n}K\Big(\frac{x_i-x_0}{h_{n,\tau}}\Big)\rho_\tau \Big( y_i- \alpha - (\beta^+_v d^+_i + \beta^-_v d^-_i)(x_i-x_0) \Big).
\end{align*}

In the remainder of this section, we obtain weak convergence results for the quantile processes of $(\hat{\beta}^+_1(\tau),\hat{\beta}^-_1(\tau))$, which in turn yield a weak convergence result for the quantile process of the QRKD estimator of treatment effects.
Using these results, we propose methods to test hypotheses concerning heterogeneous treatment effects in Section \ref{sec:test}.
Define the kernel-dependent constant matrix $N=\int \bar u \bar u' K(u)du$, where 
$\bar u=[1, ud^+_u, ud^-_u,...,u^pd^+_u, u^pd^-_u ]'\in \mathds{R}^{2p+1}$, $d^+_u=\mathds{1}\{u>0\}$ and $d^-_u=\mathds{1}\{u<0\}$.
We assume that there exist constants $\underline{x}<x_0$ and $\bar{x}>x_0$
such that the following conditions are satisfied.

\begin{assumption}\label{a:consistency}
(i) (a) The density function $f_X( \cdot )$ exists and is continuously differentiable in a neighborhood of $x_0$ and $0< f_X (x_0) <\infty$.
    (b) $\{(y_i,x_i)\}^n_{i=1}$ is an i.i.d. sample of $n$ observations of the bivariate random vector $(Y,X)$. 
(ii) (a) 
		$f_{Y|X} (Q_{Y|X}( \ \cdot \ |x_0)|x_0)$ is Lipschitz on $T$. 
		 (b) There exist finite constants $f_L>0$, $f_U>0$, and $\xi >0$, such that $f_{Y|X} (Q_{Y|X}(\tau | x)+\eta|x)$ lies between $f_L$ and $f_U$ for all $\tau \in T$, $|\eta| \le \xi$ and $x \in [\underbar{x},\bar{x}].$ 
(iii) (a) $Q_{Y|X}(\ \cdot \ |x_0)$, $\partial Q_{Y|X}(\ \cdot \ |x^+_0)/ \partial \tau$, and $\partial Q_{Y|X}(\ \cdot \ |x^-_0)/ \partial \tau$ exist and are Lipschitz continuous on $T$. 
      (b) $Q_{Y|X}(\tau| \cdot )$ is continuous at $x_0$. 
			For $v=0,1,...,p+1$, $(x,\tau) \mapsto \partial^{v} Q_{Y|X}(\tau|x)/ \partial x^{v}$ exists and is Lipschitz continuous on $\{(x,\tau)|x\in(x_0,\bar{x}], \tau \in T\}$ and  $\{(x,\tau)|x\in[\underbar{x},x_0), \tau \in T \}$. 
(iv) The kernel $K$ is compactly supported, Lipschitz, differentiable, and satisfying $K(\cdot)\ge 0$, $\int K(u)du=1$, $\int uK(u)du=0$ and $\norm{K}_\infty<\infty$. The matrix $N$ is positive definite. 
(v) The bandwidths satisfy $h_{n,\tau}=c(\tau)h_n$, where $nh_n^3 \rightarrow \infty$ and $nh^{2p+3}_n\to 0$ as $n\to \infty$, and $c(\cdot)$ is Lipschitz continuous satisfying $0 < \underline{c} \le c(\tau) \le \overline{c}  < \infty$ for all $\tau \in T.$ 
\end{assumption}

Parts (i)--(v) of this assumption correspond to Assumptions 1--5, respectively, of Qu and Yoon (2015a), adapted to our framework.
Part (i) (a) requires smoothness of the density of the running variable.
This can be interpreted as the design requirement for absence of endogenous sorting across the kink point $x_0$.
The i.i.d assumption in part (i) (b) is usually considered to be satisfied for micro data of random samples.
Part (ii) concerns about regularities of the conditional density function of $Y$ given $X$.
It requires sufficient smoothness, but does not rule out quantile regression kinks at $x_0$, which is the main crucial assumption for our identification argument.
Part (iii) concerns about regularities of the conditional quantile function of $Y$ given $X$.
Like part (ii), it does not rule out quantile regression kinks at $x_0$.
Part (iv) prescribes requirements for kernel functions to be chosen by users.
In Section \ref{sec:mc} for simulation studies, we propose an example of such a choice to satisfy this requirement.
Finally, part (v) specifies admissible rates at which the bandwidth parameters diminish as the sample size becomes large.
It obeys the standard rate for a first-order derivative estimation, but we also require its uniformity over quantiles $\tau$ in $T$.
While $nh^{2p+3}_n\to 0$ is required for a valid inference with higher order bias reduction, it is not necessary for the uniform Bahadur representation to hold.
We note that, with a bias correction of order $p > 1$, the optimal bandwidth for local linear estimation is compatible with this assumption.
Under this set of assumptions, we obtain uniform Bahadur representations for the component estimators, $\hat\beta^+ (\tau)$ and $\hat\beta^- (\tau)$, of our interest similarly to Qu and Yoon (2015a) -- see Lemma \ref{lemma:bahadur} in Appendix \ref{sec:lemma:bahadur}.

For conciseness of the statements, we write parts (ii) and (iii) of Assumption \ref{a:consistency} in terms of high-level objects, but it will be more interpretable if they were stated in terms of the structural primitives, $g$ and $f_{\varepsilon|X}$.
We provide sufficient conditions below. 
We introduce the short-hand notations
\begin{align*}
&f_1(y,x)=\int_{\partial V(y,x)} \frac{h_x(x,\epsilon)}{\norm{\nabla_\epsilon h(x,\epsilon)}} f_{\varepsilon \mid X} (\epsilon \mid x)  dH^{M-1}(\epsilon)\\
&f_2(y,x)=\int_{V(y,x)} \frac{\partial}{\partial x} f_{\varepsilon \mid X}(\epsilon \mid x) dm^M(\epsilon)\\
&f_3(y,x)=\int_{\partial V(y,x)} \frac{1}{\norm{\nabla_\epsilon h(x,\epsilon)}} f_{\varepsilon \mid X} (\epsilon \mid x) dH^{M-1}(\epsilon)
\end{align*}
where $h$, $V$ and $\partial V$ are defined in Section \ref{sec:identification}.
Lemma \ref{lemma:primitive} in Appendix \ref{sec:lemma:primitive} shows that Assumption \ref{a:primitive} stated below in terms of the structural primitives is sufficient for the aforementioned high-level conditions in parts (ii) and (iii) of Assumption \ref{a:consistency}. 
Define $y^*=\underset{(\tau,x)\in T\times [\underline x, \overline x]}{\sup} \inf\{y\in \mathcal{Y}|\int_{V(y,x)} f_{\varepsilon|X}(\epsilon|x)dm^M(\epsilon)\ge \tau \}$ and $y_*=\underset{(\tau,x)\in T\times [\underline x, \overline x]}{\inf} \inf\{y\in \mathcal{Y}|\int_{V(y,x)} f_{\varepsilon|X}(\epsilon|x)dm^M(\epsilon)\ge \tau \}$.

\begin{assumption}\label{a:primitive}
(i) $\frac{\partial^q }{\partial x^j \partial y^{q-j}}f_1$ and $\frac{\partial^q }{\partial x^j \partial y^{q-j}}f_2$ exist for each $0\le j$, $ q \le p-1$, $j+q\le p-1$ and are Lipschitz on $[y_*,y^*]\times[\underline x, x_0)$ and $[y_*,y^*]\times (x_0,\overline x]$.
(ii) $f_3$ is Lipschitz on $[y_*,y^*]\times[\underline x, \overline x]$.  $\frac{\partial^q }{\partial x^j \partial y^{q-j}}f_3$ exists for each $0\le j\le q \le p$ and is Lipschitz on $[y_*,y^*]\times[\underline x, x_0)$ and $[y_*,y^*]\times (x_0,\overline x]$. 
(iii) For each $\kappa \in (0,\infty)$, there exist finite positive constants $f'_L(\kappa)$ and $f'_U(\kappa)$ such that $0<f'_L(\kappa)< f_3(y,x)< f'_U(\kappa)<\infty$ uniformly in $(y,x)$ on $ [-\kappa,\kappa]\times [\underline x, \overline x]$.
(iv) $y_*>-\infty$ and $y^*<\infty$.
(v)  For each $\tau\in T$, $x\mapsto\inf\{y\in \mathcal{Y}|\int_{V(y,x)} f_{\varepsilon|X}(\epsilon|x)dm^M(\epsilon)\ge \tau \}$ is $p-1$-time differentiable on $[\underline x,\overline x]$. Furthermore, $\tau\mapsto\inf\{y\in \mathcal{Y}|\int_{V(y,x_0)} f_{\varepsilon|X}(\epsilon|x_0)dm^M(\epsilon)\ge \tau \}$ is Lipschitz on $T$.
(vi) $f_{\varepsilon|X}(\cdot|x_0)$ is Lipschitz.

\end{assumption} 

We now state weak convergence results for our component estimators.
\begin{theorem}\label{theorem:weakconv1}
Suppose that Assumption \ref{a:consistency} holds.
Let $\mathds{Z}_n = \mathds{Z}_n( \cdot , \cdot)$ be defined by\small
\begin{align*}
&\begin{bmatrix}
\mathds{Z}_n(\tau_1,2)\\
\mathds{Z}_n(\tau_2,3)
\end{bmatrix}
=
\\
&\begin{bmatrix}
 \sqrt{nh_{n,\tau_1}^3 }\Big( \hat{\beta}^+_1 (\tau_1)- \frac{\partial Q_{Y|X}(\tau_1|x^+_0) }{\partial x} - h^p_{n,{\tau_1}} \frac{\iota'_2 (N)^{-1}}{(p+1)!} \int_{\mathds{R}} \bar u \Big( \frac{\partial^{p+1} Q_{Y|X}(\tau_1|x^+_0)}{\partial x^{p+1}}d^+_u + \frac{\partial^{p+1} Q_{Y|X}(\tau_1|x^-_0)}{\partial x^{p+1}}d^-_u \Big)  u^{p+1}     K(u) du\Big)\\
  \sqrt{nh_{n,\tau_2}^3 }\Big( \hat{\beta}^-_1 (\tau_2)- \frac{\partial Q_{Y|X}(\tau_2|x^-_0) }{\partial x} - h^p_{n,{\tau_2}}  \frac{\iota'_3 (N)^{-1}}{(p+1)!}\int_{\mathds{R}} \bar u \Big( \frac{\partial^{p+1} Q_{Y|X}(\tau_2|x^+_0)}{\partial x^{p+1}}d^+_u + \frac{\partial^{p+1} Q_{Y|X}(\tau_2|x^-_0)}{\partial x^{p+1}}d^-_u \Big)  u^{p+1}     K(u) du\Big)
\end{bmatrix}
\end{align*}\normalsize
We have the weak convergence
$
\mathds{Z}_n\Rightarrow G 
$
for a tight zero mean Gaussian process $G:\Omega\mapsto \ell^\infty(T\times \{2,3\})$ with covariance function given by
\begin{align*}
E[G(\tau_1,j_1)G(\tau_2,j_2)]=\frac{\iota'_{j_1} N^{-1}  T(\tau_1,\tau_2) N^{-1}\iota_{j_2} (\tau_1\wedge \tau_2 -\tau_1\tau_2)}{f_X(x_0)f_{Y|X}(Q_{Y|X}(\tau_1|x_0)|x_0)f_{Y|X}(Q_{Y|X}(\tau_2|x_0)|x_0)}
\end{align*}
for each $r,s \in T$, where
$T(\tau_1,\tau_2)=(c(\tau_1)c(\tau_2))^{-1/2}\int_{\mathds{R}}\bar u(\tau_1) \bar u' (\tau_2)K(\frac{u}{c(\tau_1)})K(\frac{u}{c(\tau_2)})du$ and
$\bar u(\tau)=\big[1,\frac{u}{c(\tau)}d^+_u,\frac{u}{c(\tau)}d^-_u,...,\frac{u^p}{c^p(\tau)}d^+_u,\frac{u^p}{c^p(\tau)}d^-_u \big]'$. 
\end{theorem}

This result can be established by adapting Qu and Yoon (2015a) to our framework, and a proof is provided in Appendix \ref{sec:weakconv1}.
In this theorem, we explicitly write the $p$-th order bias terms for the purpose of emphasizing on what is the smallest order of biases.
However, this $p$-th order bias term goes away in large sample as $nh^{2p+3}_{n,\tau}$ goes to zero uniformly in $\tau \in T$ under the optimal bandwidths for local linear estimators.
In other words, this bias term can be considered to be negligible in the weak convergence result. 
The following weak convergence result for the QRKD estimator (\ref{eq:wald_estimator}) follows from Theorem \ref{theorem:weakconv1}.

\begin{corollary}\label{corollary:QRKD}
Suppose that Assumptions \ref{a:policy treatment} and \ref{a:consistency} hold.
We have
\begin{align*}
\sqrt{nh_{n,\tau}^3} \Big( \widehat{QRKD}(\tau) - QRKD(\tau) \Big)
\Rightarrow &  Y(\tau)=
\frac{G(\tau,2)-G(\tau,3)}{\left( b'(x^+_0)-b'(x^-_0) \right) }
\end{align*}
where $Y$ is a zero mean Gaussian process with covariance function
\begin{align*}
&EY(\tau_1)Y(\tau_2)=(\tau_1 \wedge \tau_2-\tau_1\tau_2)\times\\&\frac{\iota'_2 N^{-1}T(\tau_1,\tau_2)N^{-1}\iota_2+\iota'_3 N^{-1}T(\tau_1,\tau_2)N^{-1}\iota_3-\iota'_2 N^{-1}T(\tau_1,\tau_2)N^{-1}\iota_3-\iota'_3 N^{-1}T(\tau_1,\tau_2)N^{-1}\iota_2}{f_X(x_0)f_{Y|X}(Q_{Y|X}(\tau_1|x_0)|x_0)f_{Y|X}(Q_{Y|X}(\tau_2|x_0)
|x_0)(b'(x^+_0)-b'(x^-_0))^2}
\end{align*} 
for all $\tau_1$, $\tau_2\in T$.
\end{corollary}
The random process $Y(\cdot)$ has mean zero, as $G(\cdot,2)$ and $G(\cdot,3)$ do. 
In practice, we can compute its covariance structure by using the pivotal method suggested in Qu and Yoon (2015a) -- see Appendix \ref{sec:bias} for a practical guide on its implementation. 
To account for higher variance from the conditional quantiles at the localities where the conditional density is small, we may also consider the following standardized version of the weak convergence results. 
Let $\sigma^s(\tau):=\{EY^2(\tau)\}^{1/2}$, and $\widehat \sigma^s(\tau)$ be the uniformly consistent standard error estimate based on the pivotal method (Section \ref{sec:bias}). An application of Slutsky's theorem and the continuous mapping theorem to Corollary \ref{corollary:QRKD} leads to the next result.

\begin{corollary}\label{corollary:student}
Suppose that Assumptions \ref{a:policy treatment} and \ref{a:consistency} hold.
If $\sigma^s(\cdot)$ is uniformly bounded away from $0$ on $T$, then we have
\begin{align*}
\sqrt{nh_{n,\tau}^3} \Big( \frac{\widehat{QRKD}(\tau)}{\widehat\sigma^s(\tau)} - \frac{QRKD(\tau) }{\sigma^s(\tau)}\Big)
\Rightarrow & Y^{std}(\tau):=\frac{Y(\tau)}{\sigma^s(\tau)}=
\frac{G(\tau,2)-G(\tau,3)}{ \sigma^s(\tau)( b'(x^+_0)-b'(x^-_0)) }
\end{align*}
where $Y^{std}$ is a zero mean Gaussian process with covariance function
\begin{align*}
&EY^{std}(\tau_1)Y^{std}(\tau_2)=(\tau_1 \wedge \tau_2-\tau_1\tau_2)\times\\&\frac{\iota'_2 N^{-1}T(\tau_1,\tau_2)N^{-1}\iota_2+\iota'_3 N^{-1}T(\tau_1,\tau_2)N^{-1}\iota_3-\iota'_2 N^{-1}T(\tau_1,\tau_2)N^{-1}\iota_3-\iota'_3 N^{-1}T(\tau_1,\tau_2)N^{-1}\iota_2}{\sigma^s(\tau_1)\sigma^s (\tau_2)f_X(x_0)f_{Y|X}(Q_{Y|X}(\tau_1|x_0)|x_0)f_{Y|X}(Q_{Y|X}(\tau_2|x_0)
|x_0)(b'(x^+_0)-b'(x^-_0))^2}
\end{align*}
and $E(Y^{std}(\tau))^2=1$ for all $\tau$, $\tau_1$, $\tau_2\in T$.
\end{corollary}

These weak convergence results are applicable for many purposes.
They are readily applicable to computing uniform confidence bands for the QRKD.
Of particular interest may be the uniform tests regarding heterogeneous treatment effects.
We discuss them in Section \ref{sec:test}.

\subsection{Testing for Heterogeneous Treatment Effects}\label{sec:test}

Researchers are often interested in the following hypotheses regarding heterogeneous treatment effects.
\begin{eqnarray*}
\text{Treatment Significance} &H^S_0:& QRKD(\tau) = 0 \quad\text{for all } \tau \in T.
\\
\text{Treatment Heterogeneity} &H^H_0:& QRKD(\tau) = QRKD(\tau^\prime) \quad\text{for all } \tau, \tau^\prime \in T.
\end{eqnarray*}
By the result in Section \ref{sec:identification_rank_invariance}, under the case of rank invariance, these hypotheses regarding $QRKD$ are equivalent to the corresponding structural hypotheses:
\begin{eqnarray*}
&H^S_0& \Longleftrightarrow \quad g_1(b(x_0),x_0,Q_{\varepsilon | X=x_0}(\tau)) = 0 \quad\text{for all } \tau \in T.
\\
&H^H_0& \Longleftrightarrow \quad g_1(b(x_0),x_0,Q_{\varepsilon | X=x_0}(\tau)) = g_1(b(x_0),x_0,Q_{\varepsilon | X=x_0}(\tau^\prime)) \quad\text{for all } \tau, \tau^\prime \in T.
\end{eqnarray*}
Furthermore, by the result in Section \ref{sec:identification}, even under the general case without rank invariance, the hypotheses regarding $QRKD$ are logically implied by the corresponding structural hypotheses, i.e.,
\begin{eqnarray*}
&H^S_0& \Longleftarrow \quad g_1(b(x_0),x_0,\epsilon) = 0 \quad\text{for all } \epsilon \in \mathbb{R}^M.
\\
&H^H_0& \Longleftarrow \quad g_1(b(x_0),x_0,\epsilon) = g_1(b(x_0),x_0,\epsilon^\prime) \quad\text{for all }\epsilon, \epsilon^\prime \in \mathbb{R}^M.
\end{eqnarray*}
Therefore, by the contrapositive logic, a rejection of the null hypothesis $H^S_0$ implies a rejection of the structural hypothesis of uniform zero.
Likewise, a rejection of the null hypothesis $H^H_0$ implies a rejection of the structural hypothesis of homogeneity. For these logical equivalences or implications, the hypotheses $H_0^S$ and $H_0^H$ may well be of great practical interest.

Both of the two hypotheses, $H_0^S$ and $H_0^H$, are considered in Koenker and Xiao (2002), Chernozhukov and Fern\'andez-Val (2005) and Qu and Yoon (2015b), among others.
Following the approach of these preceding papers, the two hypotheses, $H^S_0$ and $H^H_0$, may be tested using the statistics
\begin{eqnarray*}
WS_n(T)&=&\sup_{\tau \in T}\sqrt{n h_{n,\tau}^3}\big|\widehat{QRKD}(\tau)\big| \qquad\text{and}
\\
WH_n(T)&=&\sup_{\tau \in T}\sqrt{n h_{n,\tau}^3}\bigg|\widehat{QRKD}(\tau)- |T|^{-1}\int_{T} \widehat{QRKD}(\tau') d\tau^\prime \bigg|,
\end{eqnarray*}
or their standardized versions
\begin{eqnarray*}
WS^{std}_n(T)&=&\sup_{\tau \in T}\sqrt{n h_{n,\tau}^3}\bigg|\frac{\widehat{QRKD}(\tau)}{\widehat\sigma^{s} (\tau)}\bigg| \qquad\text{and}
\\
WH^{std}_n(T)&=&\sup_{\tau \in T}\sqrt{n h_{n,\tau}^3}\bigg|\frac{\widehat{QRKD}(\tau)- |T|^{-1}\int_{T} \widehat{QRKD}(\tau') d\tau^\prime}{\widehat\sigma^{h}(\tau)} \bigg| \qquad,
\end{eqnarray*}
respectively, where $|T|$ denotes the length (Lebesgue measure) of interval $T \subset (0,1)$, and $\widehat \sigma^h (\tau)$ denotes the uniformly consistent standard error estimate of $\sigma^h(\tau):=\{E[\phi'_{QRKD}(Y)(\tau)]^2\}^{1/2}$ based on the pivotal method (Section \ref{sec:bias}).

For the second term in the statistic $WH_n(T)$, we could also substitute a mean RKD estimator in place of $|T|^{-1}\int_{T} \widehat{QRKD}(\tau')d\tau^\prime $. Nonetheless, we use the above definition for its convenient feature that it is written as a functional only of $\widehat{QRKD}( \cdot )$.
Consequences of Corollary \ref{corollary:QRKD} are the following asymptotic distributions of these test statistics, a proof of which is provided in Appendix \ref{sec:corollary:test}.

\begin{corollary}\label{corollary:test}
Suppose that Assumptions \ref{a:policy treatment} and \ref{a:consistency} hold.
If $\sigma^s( \ \cdot \ )$ and  $\sigma^h( \ \cdot \ )$ are bounded away from zero uniformly on $T$, then\\
(i) $WS_n(T) \Rightarrow \sup_{\tau \in T} |Y(\tau)|$ and $WS^{std}_n(T) \Rightarrow \sup_{\tau \in T} |Y(\tau)/\sigma^s(\tau)|$ under the null hypothesis $H^S_0$;\\
(ii) $WH_n(T) \Rightarrow \sup_{\tau \in T} |\phi'_{QRKD}(Y)(\tau)|$ and $WH^{std}_n(T) \Rightarrow \sup_{\tau \in T} |\phi'_{QRKD}(Y)(\tau)/\sigma^h(\tau)|$ under the null hypothesis $H^H_0$, where $\phi'_{QRKD}$ $(\lambda)(\tau)=\lambda(\tau)-|T|^{-1}\int_T \lambda(\tau^\prime) d\tau^\prime$ for all $\lambda \in \ell^\infty(T)$, the space of all bounded, measurable, real-valued functions defined on $T$.
\end{corollary}


\subsection{Covariates}\label{sec:covariates}

In empirical researches, we often face the circumstances where covariates are observed in addition to the basic variables. Under a mean regression setting, Calonico, Cattaneo, Farrell and Titiunik (2016) have investigated regression discontinuity using covariates. 
This subsection presents an extension of the QRKD baseline method and its asymptotic results to models with covariates. 
Let $\textbf{W}= (W_1,...,W_k)$ denote the covariate random vector of dimension $k \in \mathbb{N}$.
We suppose that the model is compatible with the following partial linear structure:
\begin{equation}\label{eq:covariates}
y=g(b(x),x,\epsilon)+\textbf W'\theta(\epsilon)=Q_{Y|X}(\epsilon|x)+\textbf W'\theta(\epsilon)=Q_{Y|X}(\epsilon|x,\textbf W').
\end{equation}
where $\varepsilon$ is normalized to $\varepsilon\sim Uniform(0,1)$.
We focus on this simple quantile regression representation with additive covariates and a univariate $\varepsilon$ in this section to provide a practical solution in the presence of covariates.
We could maintain the non-separability of covariates and the multi-dimensionality of $\varepsilon$ by naively extending the baseline framework, but such a naive extension would be doomed to a non-practicality in the curse of dimensionality.
For the model (\ref{eq:covariates}) which we consider, we are able to obtain the same convergence rate for the estimator as in the baseline estimator.

Adding ${\textbf W}'\gamma$ to the baseline estimator, we propose
\begin{align*}
\hat\beta^+_1(\tau)=& \iota_2 ' \underset{(\alpha, \beta^+_1, \beta^-_1,...,\beta^+_p, \beta^-_p,\gamma')' \in \mathds{R}^{1+2p+k}}{\text{argmin}} \sum_{i=1}^{n}K\Big(\frac{x_i-x_0}{h_{n,\tau}}\Big)\rho_\tau \Big( y_i- \alpha - \sum_{v=1}^{p}(\beta^+_v d^+_i + \beta^-_v d^-_i)\frac{(x_i-x_0)^v}{v!} + \textbf W'_i\gamma  \Big)\\
\hat\beta^-_1(\tau)=&\iota_3 ' \underset{(\alpha, \beta^+_1, \beta^-_1,...,\beta^+_p, \beta^-_p,\gamma')' \in \mathds{R}^{1+2p+k}}{\text{argmin}} \sum_{i=1}^{n}K\Big(\frac{x_i-x_0}{h_{n,\tau}}\Big)\rho_\tau \Big( y_i- \alpha - \sum_{v=1}^{p}(\beta^+_v d^+_i + \beta^-_v d^-_i)\frac{(x_i-x_0)^v}{v!} + \textbf W'_i\gamma\Big)
\end{align*}
With these local linear estimators, the QRKD is estimated in turn by
$$
\widehat{QRKD}_{cov}(\tau) = \frac{\hat\beta^+_1(\tau)-\hat\beta^-_1(\tau)}{b'(x_0^+)-b'(x_0^-)}.
$$
For convenience of concisely presenting assumptions and results, we introduce the following short-hand notations: $ \tilde u=[1, ud^+_u, ud^-_u,..., u^p d^+_u, u^p d^-_u, \textbf v' ]'\in \mathds{R}^{1+2p+k}$ where $\textbf v=[v_1,...,v_k]'\in \mathds{R}^{k}$, $R=\int_{\mathds{R}^{k+1}}  \tilde u \tilde u'K(u)f_{\textbf W|X}(\textbf v|x_0)dudv_1...dv_k$, and $\Gamma(\tau)=\int_{\mathds{R}^{k+1}}  \tilde u \tilde u'K(u)f_{Y|\textbf W X}(g(b(x_0),x_0,\tau)+\textbf v'\theta(\tau)|\textbf v,x_0)$ $f_{\textbf W|X}(\textbf v|x_0)dudv_1...dv_k$. 
Most of the required assumptions stated in Assumption \ref{a:consistency_cov} below are direct analogues of Assumption \ref{a:consistency}. 
Let
	$\underline y\le \underset{(\epsilon,\textbf W,x)\in T \times supp(\textbf W) \times([\underline x, x_0)\cup(x_0,\overline x]) }{\inf}$ $g(b(x),x,\epsilon)+ \textbf W'\theta(\epsilon)$ and $\overline y \ge \underset{(\epsilon,\textbf W,x)\in T \times supp(\textbf W) \times ([\underline x, x_0)\cup(x_0,\overline x])  }{\sup} g(b(x),x,\epsilon)+ \textbf W'\theta(\epsilon)$.
Consider the following conditions.

\begin{assumption}\label{a:consistency_cov}
	(i) (a) $\{(y_i,x_i, \textbf W'_i)\}^n_{i=1}$ is an i.i.d. sample of n observations of $k+2$ dimensional random vector $(Y,X, \textbf W')$. Random vector ${\textbf W}$ has a compact support. (b) $f_{\textbf W|X}$ is continuously differentiable in $x$ on $[\underline x, x_0)$ and $(x_0,\overline x]$. $f_X$ is continuously differentiable at $x_0$.
	(ii) (a) $f_{Y|\textbf W X}$ is continuous on $[\underline y, \overline y]\times supp(\textbf W)\times[\underline x, \overline x]$ and is continuously differentiable and Lipschitz on $[\underline y, \overline y]\times supp(\textbf W)\times[\underline x, x_0)$ and $[\underline y, \overline y]\times supp(\textbf W)\times( x_0, \overline x]$. (b) There exist finite constants $f_L>0$ and $f_U>0$, such that $f_{Y|\textbf WX} (g(b(x),x,\epsilon)+ \textbf W'\theta(\epsilon)+\eta|\textbf W',x)$ lies between $f_L$ and $f_U$ for all $\epsilon \in T$, $|\eta| \le \infty$ and $(x,\textbf W') \in [\underbar{x},\bar{x}]\times supp(W).$
	(iii) (a) $g(b(x_0),x_0,\epsilon)$ and $\frac{\partial}{\partial \epsilon}g(b(x_0),x_0,\epsilon)$ exist and are Lipschitz continuous in $\epsilon$ on $T$. Each coordinate of $\theta(\epsilon)$ is continuously differentiable and their derivatives are Lipschitz continuous in $\epsilon$ on $T$. 
	(b) $g(b(x),x,\epsilon)$ is continuous in $x$ at $x_0$. For $v=0,1,...,p+1$, $(x,\epsilon) \mapsto \frac{\partial^{v}}{\partial x^{v}}[g(b(x),x,\epsilon)]$ exists and is Lipschitz continuous on $\{(x,\epsilon)|x\in(x_0,\bar{x}], \epsilon \in T\}$ and  $\{(x,\epsilon)|x\in[\underbar{x},x_0), \epsilon \in T \}$. 
	(iv) The kernel $K$ is compactly supported, Lipschitz, differentiable, and satisfying $K(\cdot)\ge 0$, $\int K(u)du=1$, $\int uK(u)du=0$ and $\norm{K}_\infty<\infty$. 
	The matrices $R$ and $\Gamma(\epsilon)$ are positive definite for each $\epsilon\in T$ and the entries of their inverse matrices are uniformly bounded functions in $\epsilon\in T$. 
	(v) The bandwidths satisfy $h_{n,\epsilon}=c(\epsilon)h_n$, where $nh_n^3 \rightarrow \infty$ and $nh^{2p+3}_n\to 0$ as $n\to \infty$, and $c(\cdot)$ is Lipschitz continuous satisfying $0 < \underline{c} \le c(\epsilon) \le \overline{c}  < \infty$ for all $\epsilon \in T.$
\end{assumption}

The following theorem states weak convergence results for the model (\ref{eq:covariates}) with covariates, analogously to Theorem \ref{theorem:weakconv1} and Corollary \ref{corollary:QRKD} for the baseline model. The proofs are similar to their baseline counterparts and are therefore omitted.

\begin{theorem}\label{theorem:weak_conv_cov}
	Suppose that Assumption \ref{a:consistency_cov} holds for (\ref{eq:covariates}).
	Define $\mathds{X}''_n$ by
	\begin{align*}
	&\begin{bmatrix}
	\mathds{X}''_n(\tau_1,2)\\
	\mathds{X}''_n(\tau_2,3)
	\end{bmatrix}
	= \begin{bmatrix}
	\sqrt{nh_{n,\tau_1}^3}\Big( \hat{\beta}^+_1 (\tau_1)- \frac{\partial Q_{Y|X}(\tau_1|x^+_0) }{\partial x} \Big)  \Big)\\
	\sqrt{nh_{n,\tau_2}^3}\Big( \hat{\beta}^-_1 (\tau_2)- \frac{\partial Q_{Y|X}(\tau_2|x^-_0) }{\partial x}  \Big)  \Big)
	\end{bmatrix}.
	\end{align*}
	There exists a tight zero mean Gaussian process $G_{cov}:\Omega\mapsto\ell^\infty(T\times \{2,3\})$ with covariance function
	\begin{align*}
EG_{cov}(\tau_1,j_1)G_{cov}(\tau_2,j_2)=\frac{\iota'_{j_1}(\Gamma(\tau_1))^{-1}  \tilde T(\tau_1,\tau_2) (\Gamma(\tau_2))^{-1}\iota_{j_2} (\tau_1\wedge \tau_2 -\tau_1\tau_2)}{f_X(x_0)},
	\end{align*}
	where $\tilde T(\tau_1,\tau_2)=(c(\tau_1)c(\tau_2))^{-1/2}\int \tilde u(\tau_1) \tilde u'(\tau_2)K(u/c(\tau_1))K(u/c(\tau_2))f_{\textbf W|X}(v_1,...,v_k|x_0)dudv_1...dv_k$ and $\tilde u(\tau)=[1, ud^+_u/(\tau),ud^-_u/(\tau),...,(ud^+_u/c(\tau))^p,(ud^-_u/c(\tau))^p,v_1,...,v_k]'\in \mathds{R}^{1+2p+k}$,
	such that
$
	\mathds{X}''_n\Rightarrow G_{cov}.
$	
	Consequently, if Assumption \ref{a:policy treatment} also holds, then
	\begin{align*}
	\sqrt{nh^3_{n,\tau}}\Big(\widehat{ QRKD}_{cov}(\tau)-QRKD_{cov}(\tau)\Big)\Rightarrow Y_{cov}(\tau):=\frac{G_{cov}(\tau,2)-G_{cov}(\tau,3)}{b'(x_0^+)-b'(x_0^-)}.
	\end{align*}
\end{theorem}

\section{Simulation Studies}\label{sec:mc}

In this section, we report the performance of our causal inference methods using simulated data.
The main building blocks for the model consist of the policy function $b$, the outcome production function $g$, and the joint distribution of $(x,\varepsilon)$.
Consider the following policy function with a kink at $x_0 = 0$.
\begin{equation*}
b(x) = \begin{cases} -x & \text{if } x \leqslant 0 \\ x & \text{if } x > 0\end{cases}
\end{equation*}
For convenience of assessing the performance of our estimator for homogeneous treatment effects and heterogeneous treatment effects, we consider the following three outcome structures.
\begin{eqnarray*}
\text{Structure 0:} && g(b,x,\epsilon) = 0.0 b + 1.0 x + 0.1 x^2 + \epsilon
\\
\text{Structure 1:} && g(b,x,\epsilon) = 0.5 b + 1.0 x + 0.1 x^2 + \epsilon
\\
\text{Structure 2:} && g(b,x,\epsilon) = F_{\varepsilon | X=x_0} (\epsilon) b + 1.0 x + 0.1 x^2 + \epsilon
\end{eqnarray*}
where $F_{\varepsilon | X=x_0}$ denotes the conditional CDF of $\varepsilon$ given $X=x_0$.
Note that Structures 0 and 1 entail homogeneous treatment effects, while Structure 2 entails heterogeneous treatment effects across quantiles $\tau$ as follows.
\begin{eqnarray*}
\text{Structure 0:} && g_1(b,x,Q_{\varepsilon | X=x_0}(\tau)) = 0.0
\\
\text{Structure 1:} && g_1(b,x,Q_{\varepsilon | X=x_0}(\tau)) = 0.5
\\
\text{Structure 2:} && g_1(b,x,Q_{\varepsilon | X=x_0}(\tau)) = \tau
\end{eqnarray*}
To allow for endogeneity, we generate the primitive data according to
$$\left(\begin{array}{c} x_i \\ \varepsilon_i \end{array}\right) \stackrel{i.i.d.}{\sim} N\left(\left(\begin{array}{c}0 \\ 0\end{array}\right),\left(\begin{array}{cc}\sigma_X^2 & \rho \sigma_X \sigma_\varepsilon \\ \rho \sigma_X \sigma_\varepsilon & \sigma_\varepsilon^2 \end{array}\right)\right),$$ 
where $\sigma_X = 1.0$ and $\sigma_\varepsilon = \rho = 0.5$.
For estimation, we use the tricube kernel function $K$ defined by
\begin{equation*}
K(u) = \frac{70}{81} \left( 1 - \abs{u}^3 \right)^3 \mathbbm{1}\{\abs{u} < 1\}.
\end{equation*}
We set $p=2$, and the bandwidths are selected with the choice rule based on the MSE minimization for local linear estimator -- see Appendix \ref{sec:guide_to_practice_bandwidth} for details.

Figure \ref{fig:montecarlo_structure12} shows simulated distributions of the QRKD estimates under Structure 1 (left) and Structure 2 (right).
The top row, the middle row, and the bottom row report results for the sample sizes of $N=1,000$, $2,000$, and $4,000$, respectively.
In each graph, the horizontal axis measures quantiles $\tau$, while the vertical axis measures the QRKD.
The true QRKD is indicated by solid gray lines.
Note that it is constant at $0.5$ in the left column for Structure 1, while it is increasing in $\tau$ in the right column for Structure 2.
The other broken curves indicate the 5-th, 10-th, 50-th, 90-th, and 95-th percentiles of the simulated distributions of the QRKD estimates based on Monte Carlo 2,500 iterations.
Observe that the displayed distribution shrinks for each structure at each quantile $\tau$ as the sample size $N$ increases.
The biases appear to be minor relative to the variances, which is consistent with our employment of the bias corrected estimation approach.

In order to more quantitatively analyze the finite sample pattern, we summarize some basic statistics for the simulated distributions in Table \ref{tab:montecarlo_structure12} for Structure 1 (top panel) and Structure 2 (bottom panel).
In each panel, the three column groups list the absolute biases ($\mid$Bias$\mid$), the standard deviations (SD), and the root mean squared errors (RMSE).
For each structure at each quantile $\tau$, we again observe that SD and RMSE decrease as the sample size $N$ increases.
The biases are minor relative to the variances.
These patterns are of course consistent with our previous discussions on Figure \ref{fig:montecarlo_structure12}.

\begin{figure}
	\caption{Simulated distributions of QRKD estimates.}
	\centering
	\begin{tabular}{cc}
		Structure 1; $N=1,000$ & Structure 2; $N=1,000$ \\
		\includegraphics[width=0.40\textwidth]{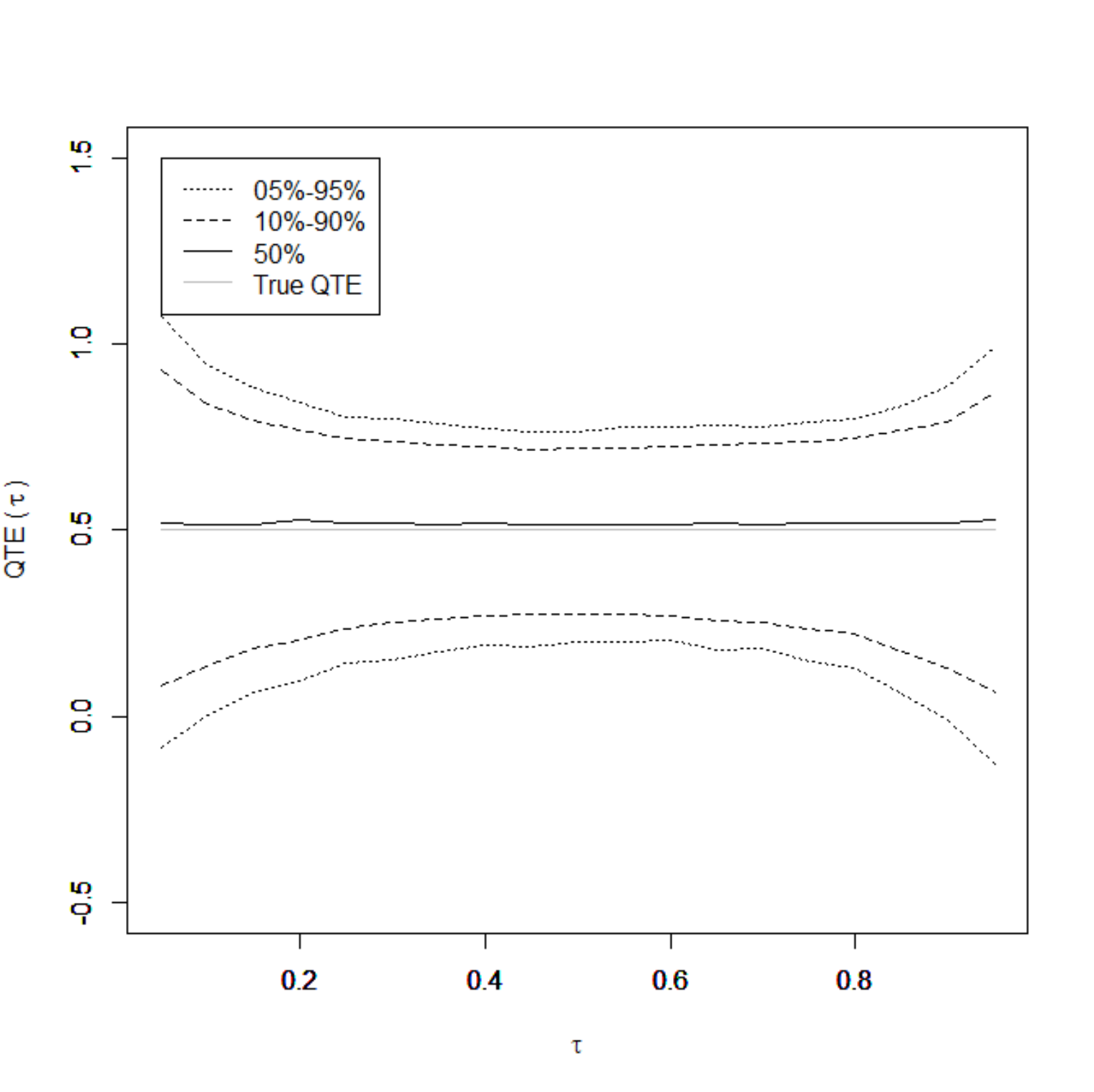} &
		\includegraphics[width=0.40\textwidth]{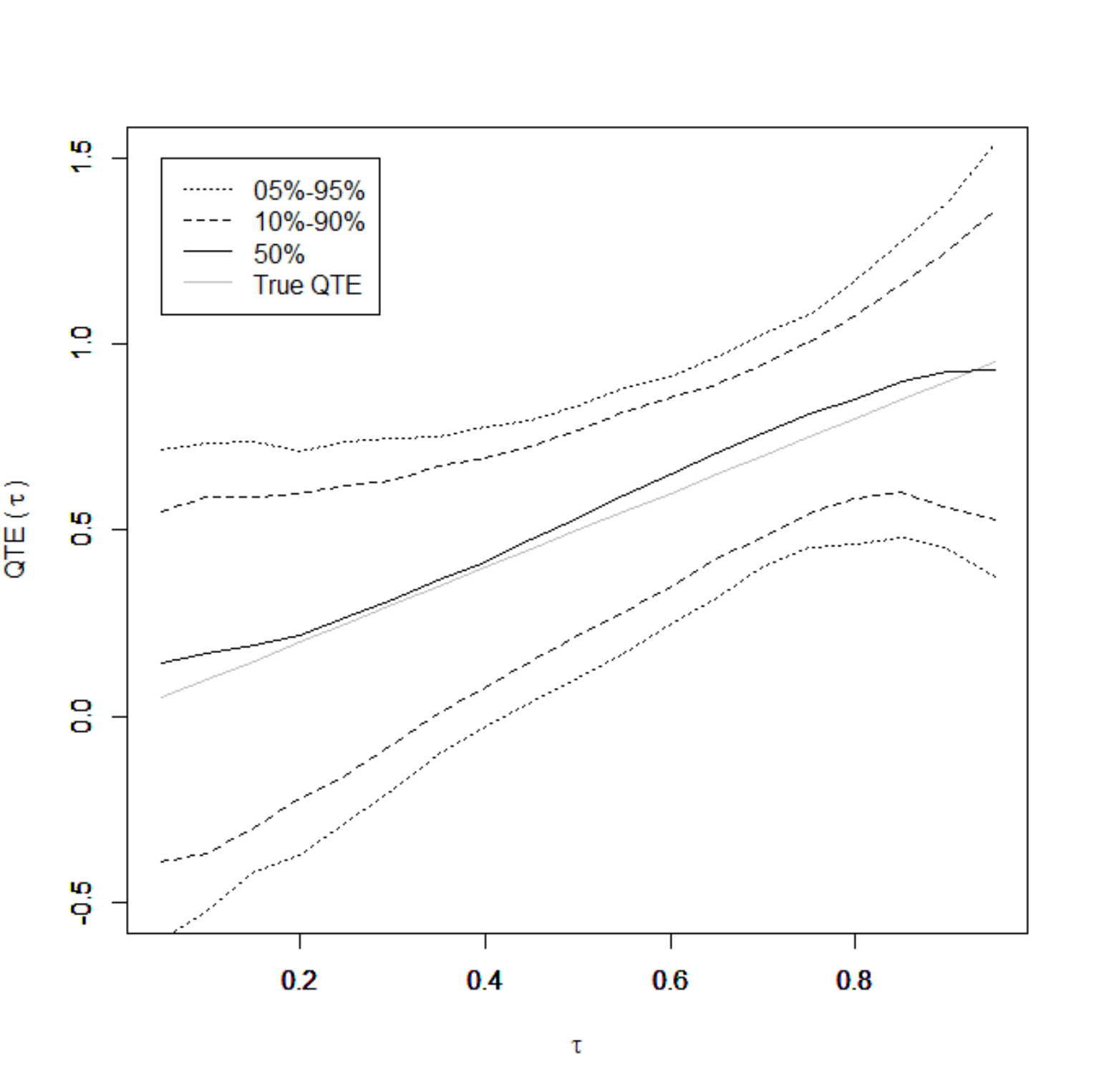}
		\\
		Structure 1; $N=2,000$ & Structure 2; $N=2,000$ \\
		\includegraphics[width=0.40\textwidth]{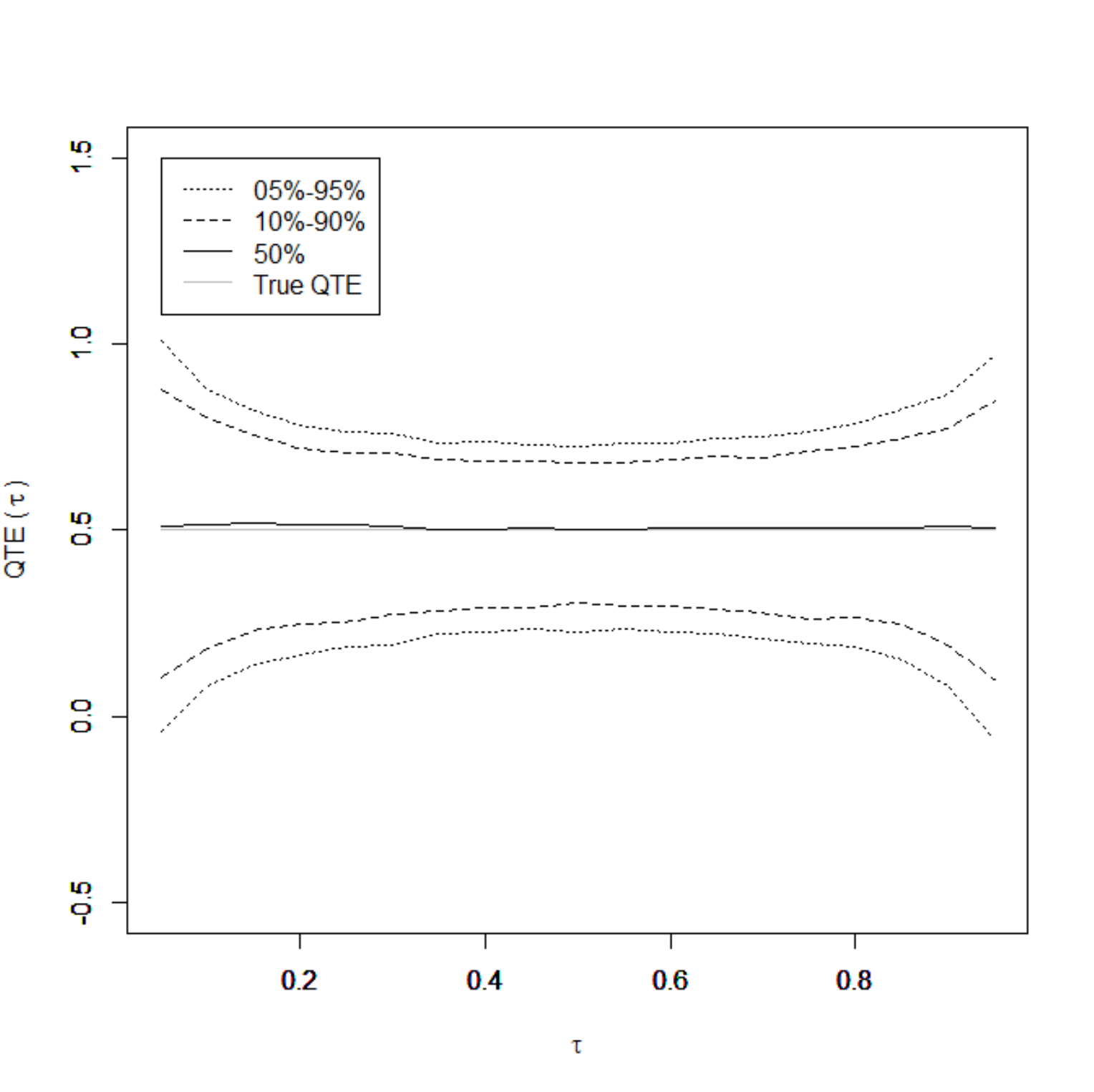} &
		\includegraphics[width=0.40\textwidth]{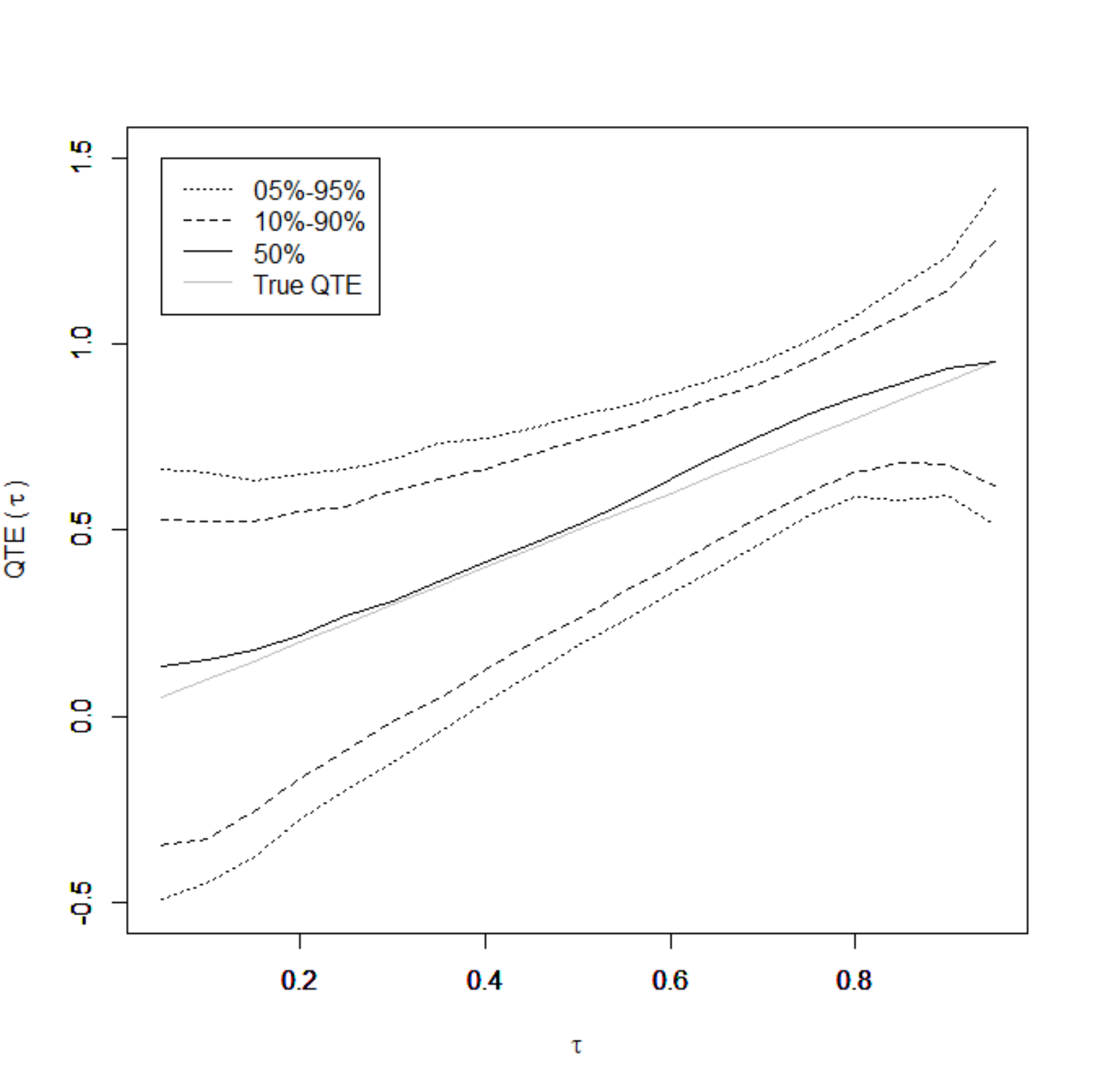}
		\\
		Structure 1; $N=4,000$ & Structure 2; $N=4,000$ \\
		\includegraphics[width=0.40\textwidth]{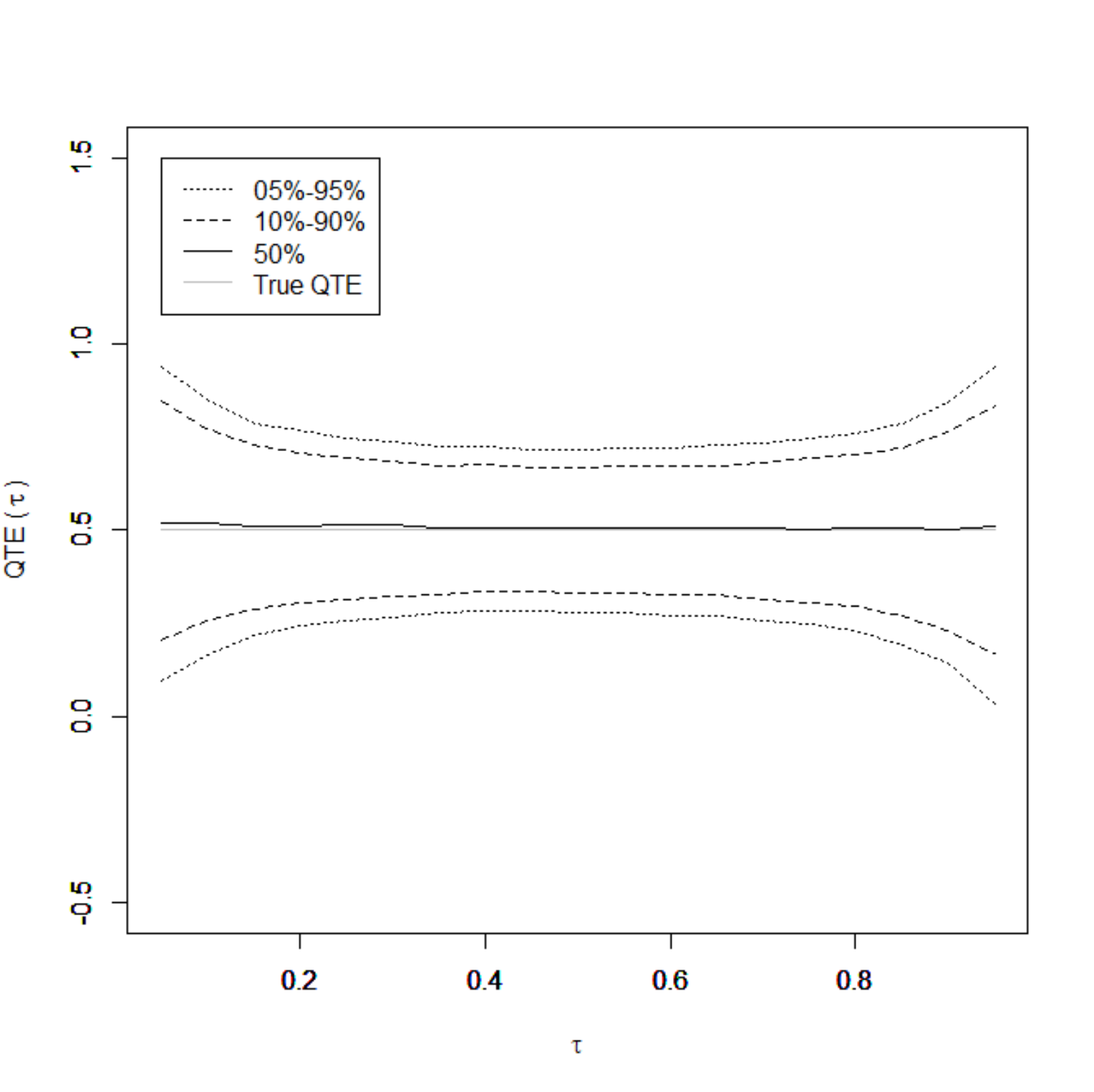} &
		\includegraphics[width=0.40\textwidth]{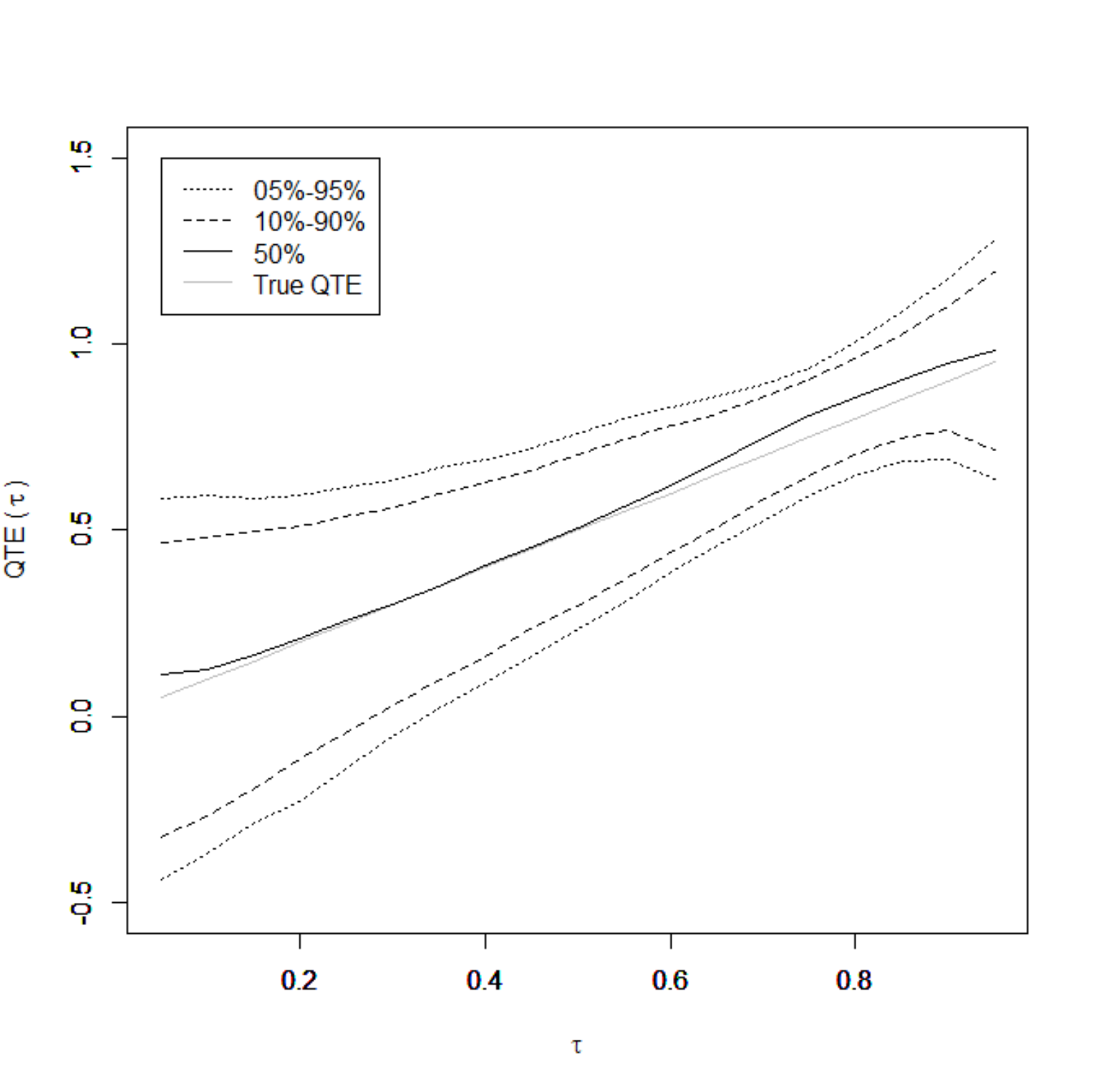}
	\end{tabular}
	\label{fig:montecarlo_structure12}
\end{figure}

\begin{table}
	\centering
		\begin{tabular}{rlllrlllrlll}
		\hline\hline
		\small Structure 1 & \multicolumn{3}{c}{$\mid$Bias$\mid$} && \multicolumn{3}{c}{SD} && \multicolumn{3}{c}{RMSE}\\
		\cline{2-4}\cline{6-8}\cline{10-12}
$N=$ &\small 1000	&\small 2000	&\small 4000	&	&\small 1000	&\small 2000	&\small 4000	&	&\small 1000	&\small 2000	&\small 4000\\
\cline{2-4}\cline{6-8}\cline{10-12}
$\tau=$ 0.10	&0.00 & 0.00 & 0.02 &  & 0.28 & 0.24 & 0.20 &  & 0.28 & 0.24 & 0.20 \\
$\tau=$ 0.20	&0.00 & 0.00 & 0.01 &  & 0.22 & 0.19 & 0.16 &  & 0.22 & 0.19 & 0.16 \\
$\tau=$ 0.30	&0.00 & 0.00 & 0.01 &  & 0.19 & 0.17 & 0.14 &  & 0.19 & 0.17 & 0.14 \\
$\tau=$ 0.40	&0.01 & 0.00 & 0.01 &  & 0.18 & 0.15 & 0.13 &  & 0.18 & 0.16 & 0.13 \\
$\tau=$ 0.50	&0.00 & 0.00 & 0.00 &  & 0.18 & 0.15 & 0.13 &  & 0.18 & 0.15 & 0.13 \\
$\tau=$ 0.60	&0.00 & 0.00 & 0.00 &  & 0.18 & 0.16 & 0.14 &  & 0.18 & 0.16 & 0.14 \\
$\tau=$ 0.70	&0.00 & 0.00 & 0.00 &  & 0.19 & 0.17 & 0.14 &  & 0.19 & 0.17 & 0.14 \\
$\tau=$ 0.80	&0.00 & 0.00 & 0.00 &  & 0.21 & 0.18 & 0.16 &  & 0.21 & 0.18 & 0.16 \\
$\tau=$ 0.90	&0.01 & 0.00 & 0.00 &  & 0.28 & 0.24 & 0.21 &  & 0.28 & 0.24 & 0.21 \\
\hline
		\small Structure 2 & \multicolumn{3}{c}{$\mid$Bias$\mid$} && \multicolumn{3}{c}{SD} && \multicolumn{3}{c}{RMSE}\\
		\cline{2-4}\cline{6-8}\cline{10-12}
$N=$ &\small 1000	&\small 2000	&\small 4000	&	&\small 1000	&\small 2000	&\small 4000	&	&\small 1000	&\small 2000	&\small 4000\\
\cline{2-4}\cline{6-8}\cline{10-12}
$\tau=$ 0.10	&0.04 & 0.03 & 0.02 &  & 0.38 & 0.34 & 0.29 &  & 0.38 & 0.34 & 0.29\\
$\tau=$ 0.20	&0.00 & 0.00 & 0.00 &  & 0.33 & 0.28 & 0.24 &  & 0.33 & 0.28 & 0.24\\
$\tau=$ 0.30	&0.00 & 0.00 & 0.00 &  & 0.28 & 0.25 & 0.21 &  & 0.28 & 0.25 & 0.21\\
$\tau=$ 0.40	&0.00 & 0.00 & 0.00 &  & 0.25 & 0.22 & 0.18 &  & 0.25 & 0.22 & 0.18\\
$\tau=$ 0.50	&0.01 & 0.01 & 0.00 &  & 0.22 & 0.19 & 0.16 &  & 0.23 & 0.19 & 0.16\\
$\tau=$ 0.60	&0.02 & 0.02 & 0.02 &  & 0.20 & 0.17 & 0.14 &  & 0.21 & 0.17 & 0.14\\
$\tau=$ 0.70	&0.04 & 0.03 & 0.03 &  & 0.19 & 0.15 & 0.11 &  & 0.19 & 0.15 & 0.12\\
$\tau=$ 0.80	&0.04 & 0.04 & 0.04 &  & 0.21 & 0.15 & 0.11 &  & 0.21 & 0.16 & 0.12\\
$\tau=$ 0.90	&0.02 & 0.02 & 0.04 &  & 0.28 & 0.20 & 0.15 &  & 0.28 & 0.20 & 0.15\\
		\hline\hline
		\end{tabular}
	\caption{Simulated finite-sample statistics of the QRKD estimates.}
	\label{tab:montecarlo_structure12}
\end{table}

Finally, we present uniform inference results using the techniques introduced in Section \ref{sec:test}.
Figure \ref{fig:mc_inference} shows acceptance probabilities for the 95\% level uniform test of significance (panel A) and the 95\% level uniform test of heterogeneity (panel B) based on 2,500 iterations.
Panel A shows that the acceptance probability for the test of the null hypothesis of insignificance converges to the nominal probability 95\% for Structure 0, while the acceptance probability decreases toward zero as the sample size increase for each of Structure 1 and Structure 2.
These results are consistent with the construction of Structure 0, Structure 1, and Structure 2.
Structure 0 exhibits the uniform zero QRKD, while neither of Structure 1 nor Structure 2 has the uniform zero QRKD.
Panel B shows that the acceptance probability for the test of the null hypothesis of homogeneity converges to the nominal probability 95\% for Structure 0 and Structure 1, while the acceptance probability decreases toward zero as the sample size increases for Structure 2.
These results are again consistent with the construction of Structure 0, Structure 1, and Structure 2.
Each of Structure 0 and Structure 1 exhibits a constant QRKD across $\tau$, while Structure 2 has non-constant QRKD across $\tau$.

\begin{figure}
	\centering
	\begin{tabular}{cc}
		\multicolumn{2}{c}{(A) Acceptance Probabilities for the 95\% Level Test of Significance}\\
		Without Standardization & With Standardization \\
		\includegraphics[width=0.5\textwidth]{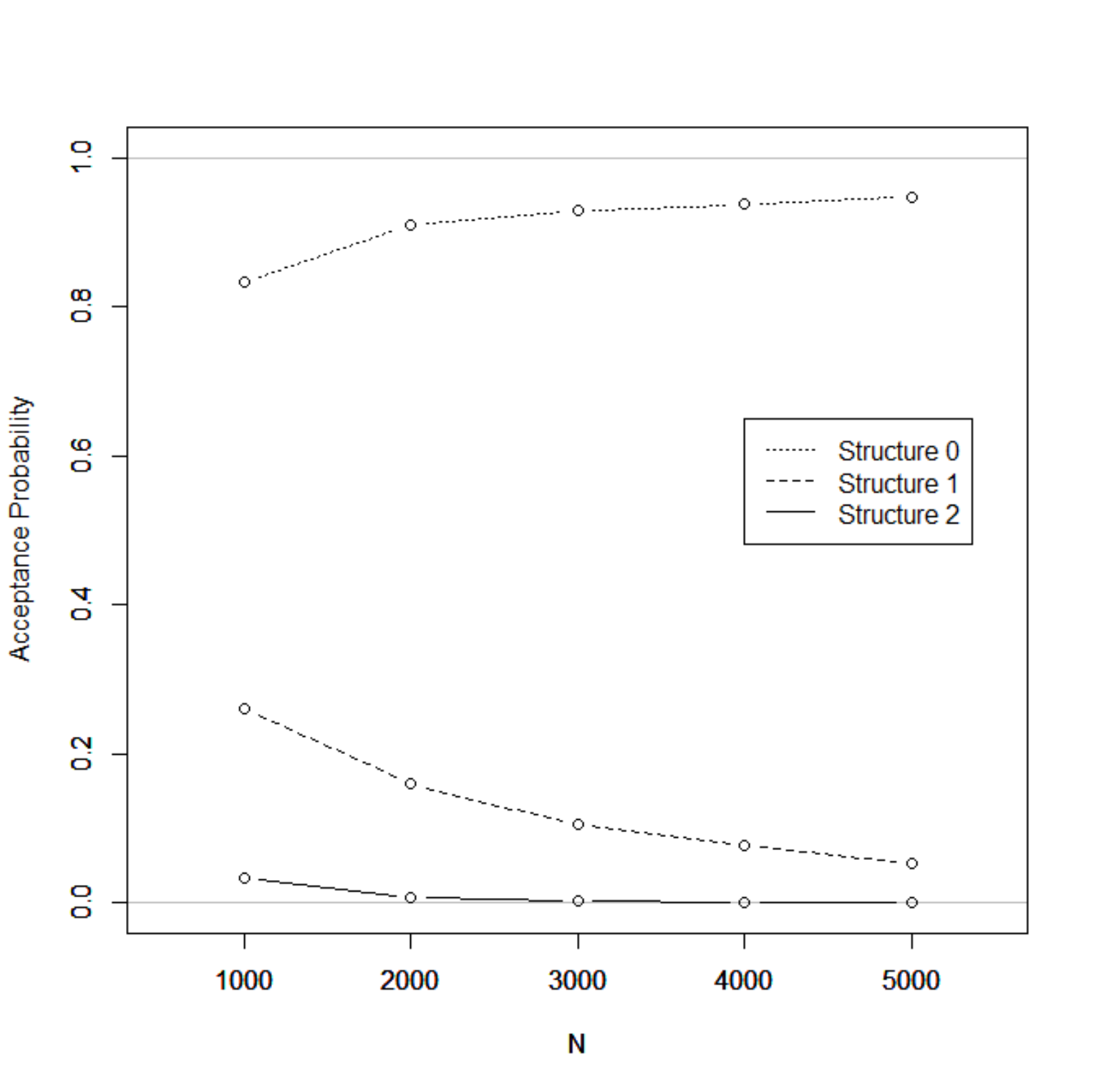} &
		\includegraphics[width=0.5\textwidth]{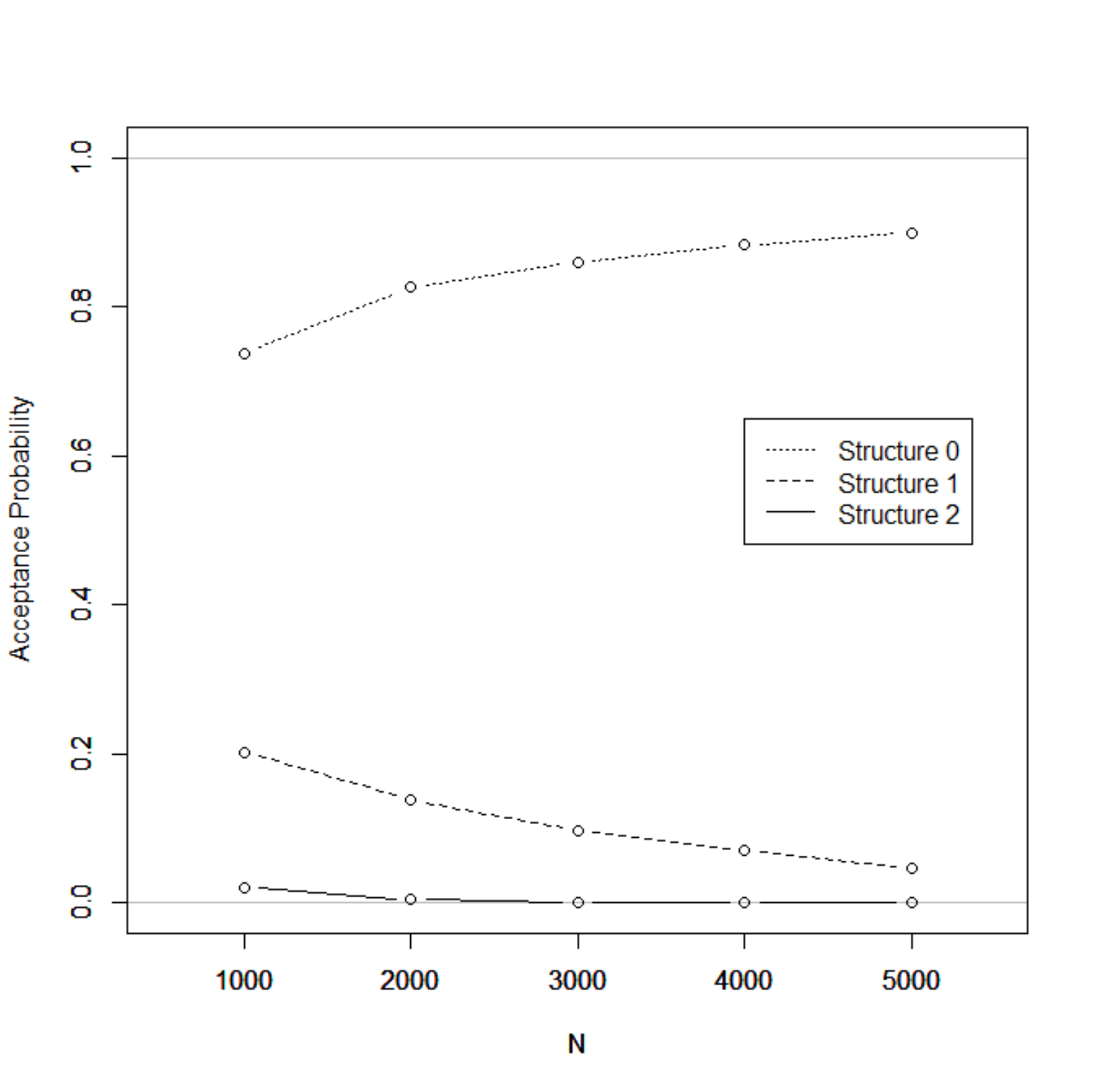} \\
		\multicolumn{2}{c}{(B) Acceptance Probabilities for the 95\% Level Test of Heterogeneity}\\
		Without Standardization & With Standardization \\
		\includegraphics[width=0.5\textwidth]{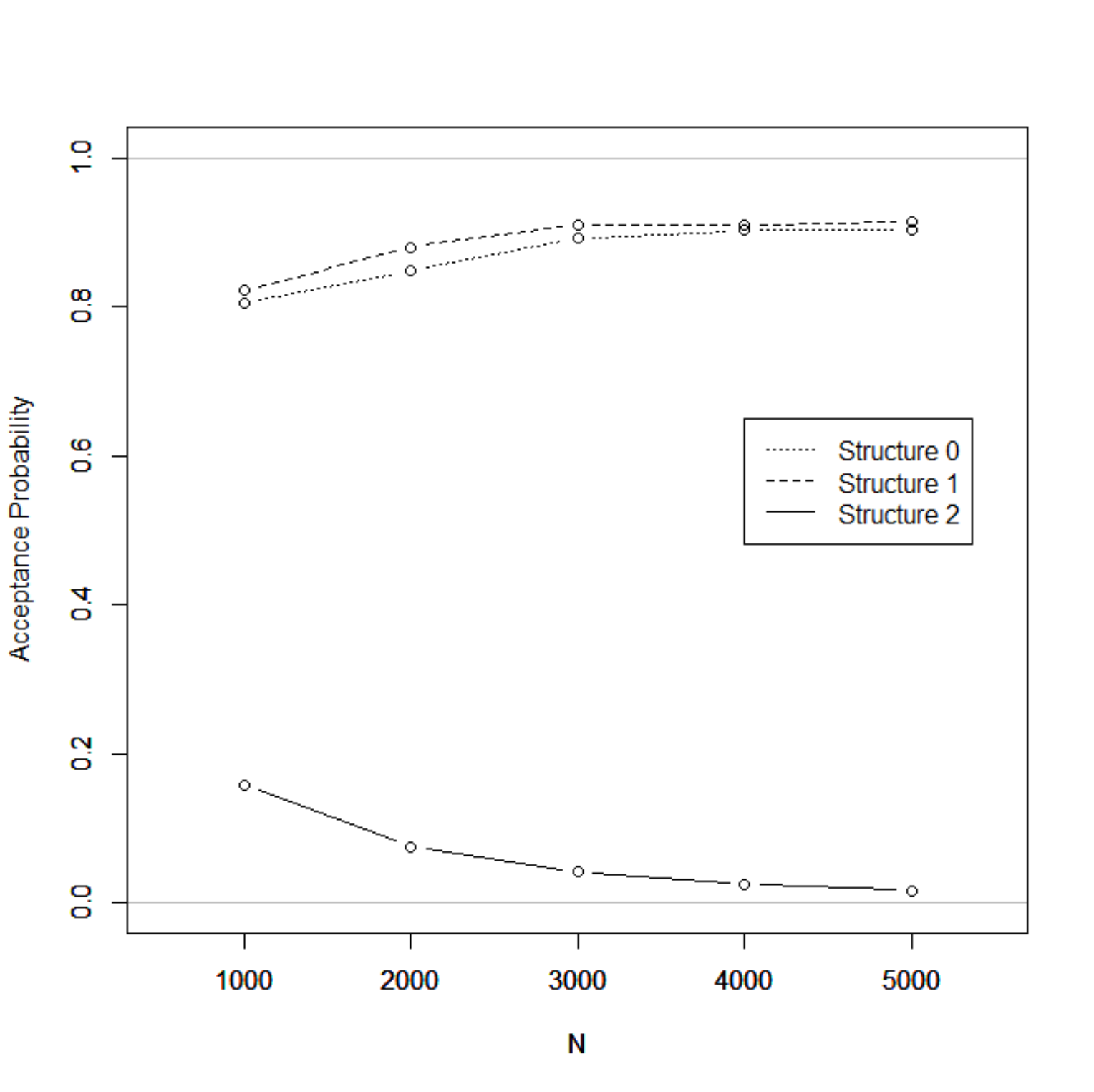} &
		\includegraphics[width=0.5\textwidth]{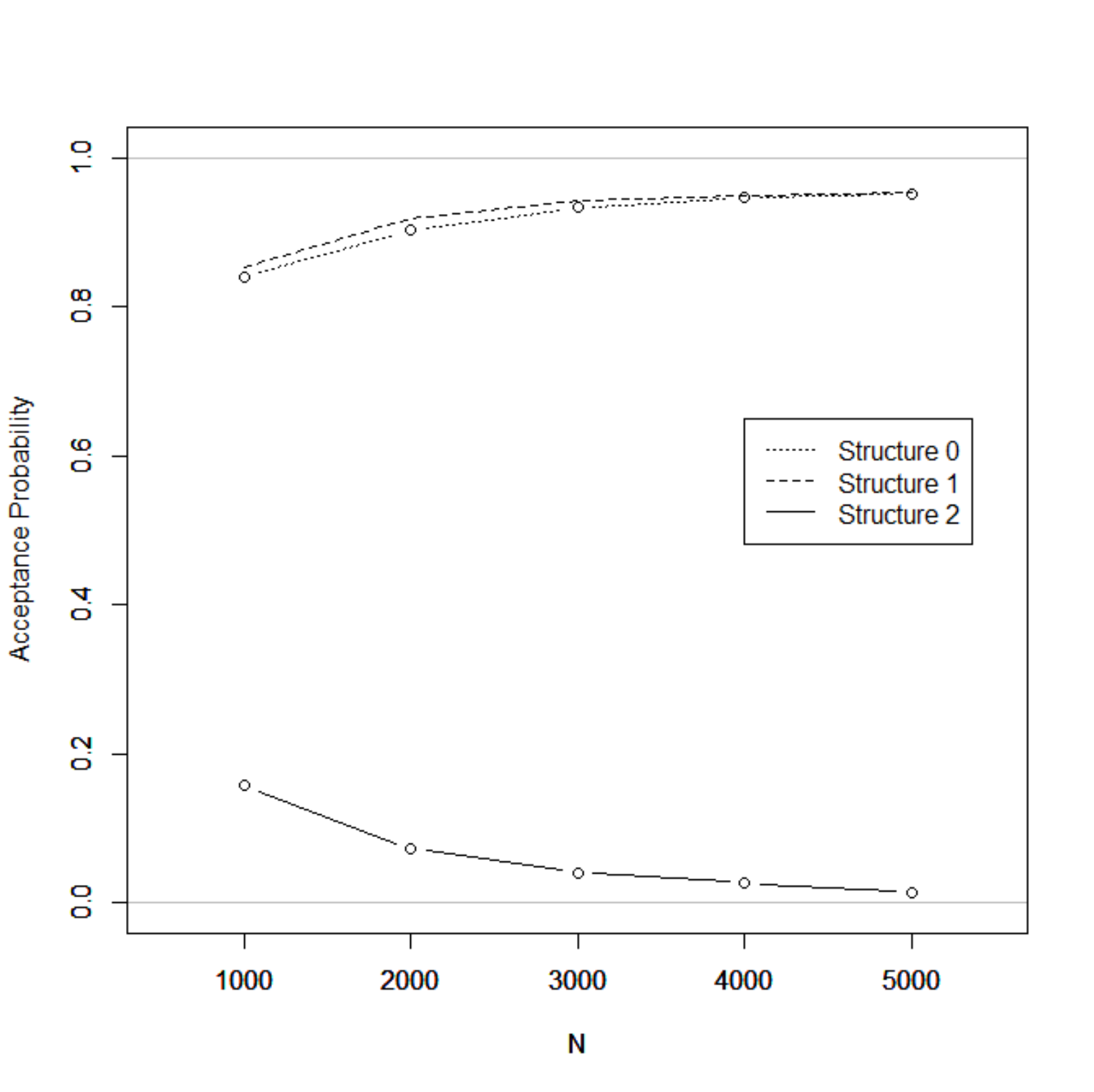}
	\end{tabular}
	\caption{Acceptance probabilities for the 95\% level uniform test of significance (panel A) and the 95\% level uniform test of heterogeneity (panel B) based on 2,500 replications.}
	\label{fig:mc_inference}
\end{figure}

\section{An Empirical Illustration}

In labor economics, causal effects of the unemployment insurance (UI) benefits on the duration of unemployment are of interest from policy perspectives.
Landais (2015) proposes an empirical strategy using the RKD to identify the causal effects of UI on the duration.
Using the data set of the Continuous Wage and Benefit History Project (CWBH -- see Moffitt, 1985), Landais estimates the effects of benefit amounts on the duration of unemployment.
In this section, we apply our QRKD methods, and aim to discover potential heterogeneity in these causal effects.
Using quantiles in this application also has an advantage of informing a likely direction of the selection bias of the mean RKD estimator that stems from not observing the mass of employed individuals at the low quantile ($y=0$).

In all of the states in the United States, a compensated unemployed individual receives a weekly benefit amount $b$ that is determined as a fraction $\tau_1$ of his or her highest earning quarter $x$ in the base period (the last four completed calendar quarters immediately preceding the start of the claim) up to a fixed maximum amount $b_{max}$, i.e. $b=\min \{\tau_1 \cdot x,\ b_{max}\}$.
The both parameters, $\tau_1$ and $b_{max}$, of the policy rule vary from state to state.
Furthermore, the ceiling level $b_{max}$ changes over time within a state.
For these reasons, empirical analysis needs to be conducted for each state for each restricted time period.
The potential duration of benefits is determined in a somewhat more complicated manner.
Yet, it also can be written as a piecewise linear and kinked function of a fraction of a running variable $x$ in the CWBH data set.

Following Landais (2015), we make our QRKD empirical illustration by using the CWBH data for Louisiana.
The data cleaning procedure is conducted in the same manner as in Landais.
As a result of the data processing, we obtain the same descriptive statistics (up to deflation) as those in Landais for those variables that we use in our analysis.
For the dependent variable $y$, we consider both the claimed number of weeks of UI and the actually paid number of weeks.
For the running variable $x$, we use the highest quarter wage in the based period.
The treatment intensity $b$ is computed by using the formula $b(x)=\min\{ (1/25) \cdot x , \ b_{max} \}$, with a kink where the maximum amount is $b_{max} = \$4,575$ for the period between September 1981 and September 1982 and $b_{max} = \$5,125$ for the period between September 1982 and December 1983.

Table \ref{tab:empirical1981} summarizes empirical results for the time period between September 1981 and September 1982.
Table \ref{tab:empirical1982} summarizes empirical results for the time period between September 1982 and December 1983.
In each table, we display the RKD results by Landais (2015) for a reference.
In the following rows, the QRKD estimates are reported with respective standard errors in parentheses for quantiles $\tau \in \{0.10, \cdots, 0.90\}$.
At the bottom of each table, we report the p-values for the test of significance and the test of heterogeneity.

\begin{table}
	\centering
		September 1981 -- September 1982
		\begin{tabular}{lccccccc}
			\hline\hline
			Dependent Variable &&& \multicolumn{2}{c}{UI Claimed} && \multicolumn{2}{c}{UI Paid} \\
			\hline
			RKD (Landais, 2015) &&& 0.038 & (0.009) && 0.040 & (0.009)\\
			\hline
			QRKD & $\tau=0.10$   && 0.000 & (0.010) && 0.022 & (0.008)\\
			     & $\tau=0.20$   && 0.037 & (0.011) && 0.036 & (0.011)\\
			     & $\tau=0.30$   && 0.053 & (0.012) && 0.060 & (0.011)\\
			     & $\tau=0.40$   && 0.070 & (0.013) && 0.070 & (0.012)\\
			     & $\tau=0.50$   && 0.081 & (0.014) && 0.080 & (0.013)\\
			     & $\tau=0.60$   && 0.093 & (0.015) && 0.089 & (0.016)\\
			     & $\tau=0.70$   && 0.086 & (0.015) && 0.068 & (0.012)\\
			     & $\tau=0.80$   && 0.154 & (0.024) && 0.142 & (0.022)\\
			     & $\tau=0.90$   && 0.145 & (0.017) && 0.159 & (0.016)\\
			\hline
			Test of Significance & $p$-Value && \multicolumn{2}{c}{0.000} && \multicolumn{2}{c}{0.000}\\
			Standardized Test of Significance & $p$-Value && \multicolumn{2}{c}{0.000} && \multicolumn{2}{c}{0.000}\\
			Test of Heterogeneity& $p$-Value && \multicolumn{2}{c}{0.000} && \multicolumn{2}{c}{0.000}\\
			Standardized Test of Heterogeneity& $p$-Value && \multicolumn{2}{c}{0.000} && \multicolumn{2}{c}{0.000}\\
			\hline\hline
		\end{tabular}
	\caption{Empirical estimates and inference for the causal effects of UI benefits on unemployment durations based on the RKD and QRKD. The period of data is from September  1981 to September 1982. The numbers in parentheses indicate standard errors.}
	\label{tab:empirical1981}
\end{table}

\begin{table}
	\centering
		September 1982 -- December 1983
		\begin{tabular}{lccccccc}
			\hline\hline
			Dependent Variable &&& \multicolumn{2}{c}{UI Claimed} && \multicolumn{2}{c}{UI Paid} \\
			\hline
			RKD (Landais, 2015) &&& 0.046 & (0.006) && 0.042 & (0.006)\\
			\hline
			QRKD & $\tau=0.10$   && 0.030 & (0.014) && 0.029 & (0.014)\\
			     & $\tau=0.20$   && 0.067 & (0.019) && 0.066 & (0.019)\\
			     & $\tau=0.30$   && 0.083 & (0.019) && 0.082 & (0.021)\\
			     & $\tau=0.40$   && 0.091 & (0.021) && 0.085 & (0.023)\\
			     & $\tau=0.50$   && 0.112 & (0.016) && 0.118 & (0.017)\\
			     & $\tau=0.60$   && 0.072 & (0.021) && 0.075 & (0.020)\\
			     & $\tau=0.70$   && 0.094 & (0.016) && 0.100 & (0.020)\\
			     & $\tau=0.80$   && 0.026 & (0.014) && 0.032 & (0.015)\\
			     & $\tau=0.90$   && 0.065 & (0.034) && 0.068 & (0.037)\\
			\hline
			Test of Significance & $p$-Value && \multicolumn{2}{c}{0.002} && \multicolumn{2}{c}{0.005}\\
			Standardized Test of Significance & $p$-Value && \multicolumn{2}{c}{0.000} && \multicolumn{2}{c}{0.000}\\
			Test of Heterogeneity& $p$-Value && \multicolumn{2}{c}{0.140} && \multicolumn{2}{c}{0.142}\\
			Standardized Test of Heterogeneity& $p$-Value && \multicolumn{2}{c}{0.000} && \multicolumn{2}{c}{0.000}\\
			\hline\hline
		\end{tabular}
	\caption{Empirical estimates and inference for the causal effects of UI benefits on unemployment durations based on the RKD and QRKD. The period of data is from September  1982 to December 1983. The numbers in parentheses indicate standard errors.}
	\label{tab:empirical1982}
\end{table}

Observe the following patterns in these result tables.
First, the estimated causal effects have positive signs throughout all the quantiles but for one ($\tau=0.10$ in Table \ref{tab:empirical1981}), implying that higher benefit amounts cause longer unemployment durations consistently across the outcome levels.
Second, these causal effects are smaller at lower quantiles (e.g., $\tau = 0.10$), while they are larger at middle and higher quantiles.
This pattern implies that unemployed individuals who have longer unemployment durations tend to have larger unemployment elasticities with respect to benefit levels.
The extent of this increase of the causal effects in quantiles is more prominent for the results in Table \ref{tab:empirical1981} (1981--1982) than in Table \ref{tab:empirical1982} (1982--1983).\footnote{We remark that the qualitative differences in the results that we find between the non-recession period (1981--1982) and the recession period (1982--1983) can be perhaps useful for telling apart the two potential routes of the causal effects, namely the moral hazard and liquidity effects.}
Under the assumption of rank invariance, this result unambiguously implies that the causal effects are heterogeneous.
Without the rank invariance, one may want to argue that the heterogeneous quantile treatment effects can be attributed to just heterogeneous weights even without nonseparable heterogeneity.
However, in the absence of nonseparability, the weights would be also constant.
Hence, our results show that there is nonseparable heterogeneity in the causal structure even without the rank invariance.
Third, the causal effects are very similar between the results for claimed UI as the outcome and the results for paid UI as the outcome variable.
The respective standard errors are almost the same between these two outcome variables, but they are not exactly the same.
Fourth, the uniform tests show that the causal effects are significantly different from zero for the both time periods.
Lastly, the uniform tests show that the causal effects are also significantly heterogeneous for the both time periods.
Indeed, the heterogeneity is insignificant in Table \ref{tab:empirical1982} (1982--1983) according to the non-standardized test statistics, but it is significant according to the standardized ones.

\section{Summary}

Economists have taken advantage of policy irregularities to assess causal effects of endogenous treatment intensities.
A new approach along this line is the regression kink design (RKD) used by recent empirical papers, including Nielsen, S\o rensen and Taber (2010), Landais (2015), Simonsen, Skipper and Skipper (2015), Card, Lee, Pei and Weber (2016), and Dong (2016).
While the prototypical framework is only able to assess the average treatment effect at the kink point, inference for heterogeneous treatment effects using the RKD is of potential interest by empirical researchers (e.g., Landais (2011) considers it).
In this light, this paper develops causal analysis and methods of inference for the quantile regression kink design (QRKD).

We first develop causal interpretations of the QRKD estimand.
It is shown that the QRKD estimand measures the marginal effect of the treatment variable on the outcome variable at the conditional quantile of the outcome given the design point of the running variable provided that the causal structure exhibits rank invariance.
This result is generalized to the case of no rank invariance, where the QRKD estimand is shown to measure a weighted average of the marginal effects of the treatment variable on the outcome variable at the conditional quantile of the outcome given the design point of the running variable. 
Second, we propose a sample counterpart QRKD estimator, and develop its asymptotic properties for statistical inference of heterogeneous treatment effects. Under some extra assumptions, a variation of the QRKD estimand that accounts for covariates is also provided. 
We obtain weak convergence results for the QRKD estimators.
Applying the weak convergence results, we propose procedures for statistical tests of treatment significance and treatment heterogeneity.
Simulation studies support our theoretical results.
Applying our methods to the Continuous Wage and Benefit History Project (CWBH) data, we find significantly heterogeneous causal effects of unemployment insurance benefits on unemployment durations in the state of Louisiana for the period between September 1981 and December 1983.


\appendix
\section{Mathematical Appendix}\label{sec:math_appendix}

In this appendix, we abbreviate $Q_{Y|X}$ to $Q$.
We also use short-hand notations $d^+_u=\mathds{1}\{u>0\}$, $d^-_u=\mathds{1}\{u<0\}$, $K_{i,n,\tau}=K(\frac{x_i-x_0}{h_{n,\tau}})$, and
$
z_{i,n,\tau}=\Big[1,(\frac{x_i-x_0}{h_{n,\tau}})d^+_i,(\frac{x_i-x_0}{h_{n,\tau}})d^-_i,...,(\frac{x_i-x_0}{h_{n,\tau}})^p d^+_i,(\frac{x_i-x_0}{h_{n,\tau}})^p d^-_i \Big]'.
$

\subsection{Uniform Bahadur Representation}\label{sec:lemma:bahadur}

The following lemma states the uniform Bahadur representation from Qu and Yoon (2015a, Theorem 1) adapted to our framework.

\begin{lemma}\label{lemma:bahadur}
Under Assumption \ref{a:consistency}, we have
\begin{align*}
&\sqrt{nh^3_{n,\tau}}\Big( \hat \beta^+_1(\tau) - \frac{\partial Q(\tau|x^+_0)}{\partial x}  - h^p_{n,\tau}  \frac{\iota'_2 (N)^{-1}}{(p+1)!} \int_{\mathds{R}} \Big( \frac{\partial^{p+1} Q(\tau|x^+_0)}{\partial x^{p+1}}d^+_u + \frac{\partial^{p+1} Q(\tau|x^-_0)}{\partial x^{p+1}}d^-_u \Big)   u^{p+1}    \bar u K(u) du
\Big)\\
=&\frac{\iota'_2 N^{-1} \sum_{i=1}^{n} z_{i,n,\tau} K_{i,n,\tau} (\tau - \mathds{1}\{ y_i \le Q(\tau|x_i) \}) }{\sqrt{nh_{n,\tau}}f_X(x_0)f_{Y|X}(Q(\tau|x_0)|x_0)} +o_p(1) \qquad\text{and}\\
&\sqrt{nh^3_{n,\tau}}\Big( \hat \beta^-_1(\tau) - \frac{\partial Q(\tau|x^-_0)}{\partial x} 
- h^p_{n,\tau}  \frac{\iota'_3 (N)^{-1}}{(p+1)!} \int_{\mathds{R}} \Big( \frac{\partial^{p+1} Q(\tau|x^+_0)}{\partial x^{p+1}}d^+_u + \frac{\partial^{p+1} Q(\tau|x^-_0)}{\partial x^{p+1}}d^-_u \Big) u^{p+1}    \bar u K(u) du\Big)\\
=&\frac{\iota'_3 N^{-1} \sum_{i=1}^{n} z_{i,n,\tau} K_{i,n,\tau} (\tau - \mathds{1}\{ y_i \le Q(\tau|x_i) \}) }{\sqrt{nh_{n,\tau}}f_X(x_0)f_{Y|X}(Q(\tau|x_0)|x_0)}+o_p(1)
\end{align*}
uniformly in $\tau \in T.$
\end{lemma}

\subsection{Stochastic Equicontinuity}\label{sec:lemma:aux_lem}

We state the stochastic equicontinuity lemma by Qu and Yoon (2015a, Lemma B.3) adapted to our framework.
Define the linear extrapolation error and the estimation error by
\begin{align*}
&e_i(\tau)=\Big[ Q(\tau|x_0)+ \sum_{v=1}^{p}\frac{(x_i-x_0)^v}{v!}\Big( \frac{\partial^v Q(\tau|x^+_0) }{\partial x^v}d^+_i + \frac{\partial^v Q(\tau|x^-_0) }{\partial x^v}d^-_i \Big)\Big] - Q(\tau|x_i) 
\qquad\text{and}\\
&\phi(\tau)=
\begin{array}{ccccc}
\sqrt{nh_{n,\tau}} \Big[ &
\alpha(\tau)- Q(\tau|x_0), &
h_{n,\tau}\Big( \beta^+_1(\tau)  - \frac{ \partial Q(\tau|x^+_0) }{\partial x} 
 \Big),&
h_{n,\tau}\Big(\beta^-_1(\tau)- \frac{ \partial Q(\tau|x^-_0) }{\partial x}\Big),
\\
&\cdots,&
(h^p_{n,\tau}/p!)\Big( \beta^+_p(\tau)   - \frac{ \partial^p Q(\tau|x^+_0) }{\partial x^p}  \Big),&
(h^p_{n,\tau}/p!)\Big( \beta^-_p(\tau) - \frac{ \partial^p Q(\tau|x^-_0) }{\partial x^p}\Big)
&
\Big],
\end{array}
\end{align*}
respectively.
The level estimator is denoted by
$$
\hat{\alpha}(\tau) = \iota_1 ' \underset{(\alpha, \beta^+_1, \beta^-_1,...,\beta^+_p, \beta^-_p)\in \mathds{R}^{1+2p}}{\text{argmin}} \sum_{i=1}^{n}K\Big(\frac{x_i-x_0}{h_{n,\tau}}\Big)\rho_\tau \Big( y_i- \alpha - \sum_{v=1}^{p}(\beta^+_v d^+_i + \beta^-_v d^-_i)\frac{(x_i-x_0)^v}{v!} \Big),
$$
where $\iota_1=[1,0,...,0]'\in \mathds{R}^{2p+1}$.
With these notations, we further define
\begin{align*}
S_n(\tau,\phi (\tau), e_i (\tau))&=\frac{1}{\sqrt{nh_n}} \sum_{i=1}^{n} \Big[ P((u^0_i (\tau) \le e_i (\tau)+(nh_{n,\tau} )^{-1/2} z'_{i,n,\tau} \phi (\tau) )|x_i) \\
&\qquad\qquad\qquad\qquad\qquad\qquad- \mathds{1}(u^0_i (\tau) \le e_i (\tau)+(nh_{n,\tau} )^{-1/2} z'_{i,n,\tau} \phi (\tau) ) \Big]z_{i,n,\tau} K_{i,n,\tau}.
\end{align*}
The following lemma corresponds to Lemma B.3 of Qu and Yoon (2015a) adapted to our framework. 
\begin{lemma}\label{lemma:b3}
Under Assumption \ref{a:consistency} (i)--(v), for any $\xi>0$ and $\eta>0$, there exists $\delta>0$ such that for $n$ large enough
$$
P\Big( \underset{\tau'',\tau' \in T, |\tau''-\tau'|\le \delta}{\sup} \norm{S_n(\tau'',0,0)-S_n(\tau',0,0)} \ge \xi \Big) <\eta.
$$
\end{lemma}

\subsection{Proof of Theorem \ref{theorem:weakconv1}}\label{sec:weakconv1}
\begin{proof}
Lemma \ref{lemma:bahadur} and Assumption \ref{a:consistency} (iii) (b), (v) imply $\mathds{Z}_n=\mathds{X}_n+o_p(1)$ uniformly in $(\tau,j)$, where
\begin{align*}
\mathds{X}_n(\tau,j)=\sum_{i=1}^{n}\frac{\iota'_j N^{-1}  z_{i,n,\tau} K_{i,n,\tau} (\tau - \mathds{1}\{ y_i \le Q(\tau|x_i) \}) }{\sqrt{nh_{n,\tau}}f_X(x_0)f_{Y|X}(Q(\tau|x_0)|x_0)}.
\end{align*}
Therefore, in light of Lemma 18.10 (iv), Theorem 18.14, and Lemma 18.15 of van der Vaart (1998), it suffices to show that the leading term in the uniform Bahadur representation $\mathds{X}_n(\tau,j)$ is asymptotically tight and has finite dimensional convergence in distribution to the Gaussian distribution with the proposed covariance function. 

The finite dimensional convergence follows from the multivariate CLT (van der Vaart, 1998, Example 2.18). For any fixed set of quantiles $T'=\{\tau_1,...,\tau_k\}\subset T$ and any $\tau\in T'$, the law of iterated expectations gives $E[\mathds{X}_n(\tau,j)]=0$ under Assumption \ref{a:consistency} (i), (ii)(b), (iv) and (v). For calculation of the covariance function, pick any $\tau_1$, $\tau_2\in T'$ and $j_1$, $j_2\in \{2,3\}$. The law of iterated expectations under Assumption \ref{a:consistency} (i), (ii)(b), (iv), and (v) implies that
\begin{align*}
&E\mathds{X}_n(\tau_1,j_1)\mathds{X}_n(\tau_2,j_2)\\
=&E\sum_{i=1}^{n}\frac{\iota'_{j_1} N^{-1}  z_{i,n,\tau_1}z'_{i,n,\tau_2} N^{-1}\iota_{j_2}K_{i,n,\tau_1}K_{i,n,\tau_2} (\tau_1 - \mathds{1}\{ y_i \le Q(\tau_1|x_i) \})(\tau_2 - \mathds{1}\{ y_i \le Q(\tau_2|x_i) \}) }{nh_n\sqrt{c(\tau_1)c(\tau_2)}f^2_X(x_0)f_{Y|X}(Q(\tau_1|x_0)|x_0)f_{Y|X}(Q(\tau_2|x_0)|x_0)}\\
=&E\sum_{i=1}^{n}\frac{\iota'_{j_1} N^{-1}  z_{i,n,\tau_1}z'_{i,n,\tau_2} N^{-1}\iota_{j_2}K_{i,n,\tau_1}K_{i,n,\tau_2} E[(\tau_1 - \mathds{1}\{ y_i \le Q(\tau_1|x_i) \})(\tau_2 - \mathds{1}\{ y_i \le Q(\tau_2|x_i) \}) |x_i]}{nh_n\sqrt{c(\tau_1)c(\tau_2)}f^2_X(x_0)f_{Y|X}(Q(\tau_1|x_0)|x_0)f_{Y|X}(Q(\tau_2|x_0)|x_0)}\\
=&E\frac{\iota'_{j_1} N^{-1}  z_{i,n,\tau_1}z'_{i,n,\tau_2} N^{-1}\iota_{j_2}K_{i,n,\tau_1}K_{i,n,\tau_2} (\tau_1\wedge \tau_2 -\tau_1\tau_2)}{h_n\sqrt{c(\tau_1)c(\tau_2)}f^2_X(x_0)f_{Y|X}(Q(\tau_1|x_0)|x_0)f_{Y|X}(Q(\tau_2|x_0)|x_0)}\\
=&\frac{\iota'_{j_1} N^{-1}  T(\tau_1,\tau_2) N^{-1}\iota_{j_2} (\tau_1\wedge \tau_2 -\tau_1\tau_2)}{f_X(x_0)f_{Y|X}(Q(\tau_1|x_0)|x_0)f_{Y|X}(Q(\tau_2|x_0)|x_0)}+o(h_n)=O(1).
\end{align*}
This establishes the finite dimensional asymptotic normality.
The tightness follows because the denominator is bounded away from zero by Assumption \ref{a:consistency} (i), (ii), and because the numerator is tight by Lemma \ref{lemma:b3}.
\end{proof}
\subsection{Proof of Corollary \ref{corollary:QRKD}}
\label{sec:corollary:QRKD}
The result follows from an application of the delta method under Assumption \ref{a:policy treatment}. 
That the limiting distribution in Theorem \ref{theorem:weakconv1} is zero-mean Gaussian implies that the limiting distribution $\frac{G(\tau,2)-G(\tau,3)}{b'(x^+_0)-b'(x^-_0)}$ is also zero-mean Gaussian. 
The covariance is obtained by
\begin{align*}
&E\Big[\Big(\frac{G(\tau_1,2)-G(\tau_1,3)}{b'(x^+_0)-b'(x^-_0)}\Big)\Big(\frac{G(\tau_2,2)-G(\tau_2,3)}{b'(x^+_0)-b'(x^-_0)}\Big)\Big]\\
=&\frac{1}{(b'(x^+_0)-b'(x^-_0))^2}E[G(\tau_1,2)G(\tau_2,2)+G(\tau_1,3)G(\tau_2,3)-G(\tau_1,2)G(\tau_2,3)-G(\tau_1,3)G(\tau_2,2)]\\
=&\frac{(\tau_1 \wedge \tau_2-\tau_1\tau_2)}{(b'(x^+_0)-b'(x^-_0))^2} \times
\\& \frac{\iota'_2 N^{-1}T(\tau_1,\tau_2)N^{-1}\iota_2+\iota'_3 N^{-1}T(\tau_1,\tau_2)N^{-1}\iota_3-\iota'_2 N^{-1}T(\tau_1,\tau_2)N^{-1}\iota_3-\iota'_3 N^{-1}T(\tau_1,\tau_2)N^{-1}\iota_2}{f_X(x_0)f_{Y|X}(Q_{Y|X}(\tau_1|x_0)|x_0)f_{Y|X}(Q_{Y|X}(\tau_2|x_0)
|x_0)}
\end{align*}
for each $\tau_1$, $\tau_2\in T$,
where the last equality follows from the covariance expression of $G$ derived in Theorem \ref{theorem:weakconv1}.

\subsection{Proof of Corollary \ref{corollary:test}}\label{sec:corollary:test}
\begin{proof}
We focus on the non-standardized tests since results for the standardized tests will follow from those for the standardized ones through Slutsky's Theorem under the stated assumptions that $\sigma^s$, $\sigma^h$ are bounded away from zero uniformly on $T$. 
Part (i) of the corollary follows straightforwardly from Corollaries \ref{corollary:QRKD} and \ref{corollary:student}.
Part (ii) of the corollary follows by an application of the functional delta method (van der Vaart,1998; Theorem 20.8) with Corollaries \ref{corollary:QRKD} and \ref{corollary:student}.
It suffices to show that the linear functional $\phi: g \mapsto g-|T|^{-1}\int_T g d\tau $ is Hadamard differentiable at $QRKD$ tangentially to $\ell^\infty(T)$.
The linearity of $\phi'_{QRKD}$ is obvious, and the continuity is implied by its boundedness as $\norm{\phi'_{QRKD}(g)} \le \norm{g} \cdot |1+diam(T)|$ for all $g \in \ell^\infty(T)$.
Let $\{g_n\}_n \subset \ell^\infty(T)$ be a sequence converging to $g \in \ell^\infty(T)$ and $t_n \to 0$.
We have
$$
\frac{\phi(QRKD+t_n g_n)-\phi(QRKD)}{t_n}-\phi'_{QRKD}(g) \to 0 \qquad\text{in } \ell^\infty(T)
$$
by the bounded convergence theorem.
This shows the Hadamard differentiability.
\end{proof}

\subsection{Sufficient Primitive Conditions}\label{sec:lemma:primitive}

In this section, we show that the primitive conditions stated in Assumption \ref{a:primitive} are sufficient for the high-level statements in parts (ii) and (iii) of Assumption \ref{a:consistency}.

\begin{lemma}\label{lemma:primitive}
Assumptions \ref{a:regularity}, \ref{a:boundary}, \ref{a:pq}, \ref{a:DCT}, and \ref{a:primitive} imply parts (ii) and (iii) of Assumption \ref{a:consistency}.
\end{lemma}

\begin{proof}
	Following the proof of Lemma 1 of Sasaki (2015) under Assumptions \ref{a:regularity}, \ref{a:boundary}, \ref{a:pq}, and \ref{a:DCT}, we obtain constants $c_j$, $j=1$, $2$, $3$, such that $c_3 \neq 0$,
	\begin{align}
	&\frac{\partial }{\partial x}Q(\tau|x)=-\frac{\frac{\partial}{\partial x}F_{Y|X}(Q(\tau|x)|x)}{f_{Y|X}(Q(\tau|x)|x)}=-\frac{c_1f_1(Q(\tau|x),x)-c_2f_2(Q(\tau|x),x)}{c_3f_3(Q(\tau|x),x)},
	\qquad\text{and}
	\label{eq:dqdx}\\
	&\frac{\partial }{\partial \tau}Q(\tau|x)=\frac{1}{f_{Y|X}(Q(\tau|x)|x)}=\frac{1}{c_3f_3(Q(\tau|x),x)}.\label{eq:dqdt}
	\end{align}
	Assumption \ref{a:primitive} (iv) implies $\underset{(\tau,x)\in T\times [\underline x, \overline x]}{\sup} |Q(\tau|x)|<\infty$.
	Assumption \ref{a:primitive} (iii) allows us to pick $\kappa$ large enough to ensure that the denominator $f_3(Q(\tau|x),x)$ is uniformly bounded away from zero.
	Using these calculations, we argue that Assumption \ref{a:primitive} implies parts (ii) and (iii) of Assumption \ref{a:consistency}.
	
	First, as in the calculation for (\ref{eq:dqdt}) above under Assumptions  \ref{a:regularity}, \ref{a:boundary}, \ref{a:pq}, and \ref{a:DCT}, we can write
\begin{align*}
f_{Y|X}(Q_{Y|X}( \ \cdot \ | x_0) | x_0)=c_3 f_3\Big(\inf\Big\{y\in \mathcal{Y}\Big|\int_{V(y,x_0)} f_{\varepsilon|X}(\epsilon|x_0)dm^M(\epsilon)\ge \tau \Big\},x_0\Big).
\end{align*}
By Assumption \ref{a:primitive} (v) and (vi), $f_{Y|X}(Q_{Y|X}( \ \cdot \ | x_0) | x_0)$ is Lipschitz on $T$.
This shows that Assumption \ref{a:consistency} (ii) (a) holds.

	Second, 
set $\kappa=\max\{|y_*|,|y^*|\}+\delta$ for a small $\delta>0$ and invoke Assumption \ref{a:primitive} (iii), (iv), so we have $0<f'_L(\kappa)<f_{Y|X}(y|x)/c_3=f_3(y,x)<f'_U(\kappa)<\infty$ uniformly in $(y,x)$ on $[-\kappa,\kappa]\times[\underline x, \overline x]$. By definitions of $\kappa$, $y_*$, and $y^*$, it holds that $-\kappa <y_*-\delta/2<y_*\le Q(\tau|x) \le y^*<y^*+\delta/2<\kappa$ on $T\times [\underline x, \overline x]$. Take $\xi=\delta/2$, and we have $f_L(\kappa) \leq f_{Y|X}(Q_{Y|X}(\tau|x)+\eta |x) \leq f_U(\kappa)$ for all $\tau \in T$, $\abs{\eta} \leq \xi$ and $x \in [\underbar{x},\bar{x}]$. This shows that Assumption \ref{a:consistency} (ii) (b) holds.
	
	Third, 
Assumption \ref{a:primitive} (v) implies that $Q(\tau|x_0)$ is Lipschitz. For $\frac{\partial}{\partial \tau}Q(\tau|x^+_0)$,  
since $\lim_{x\to x^+_0}f_3(\cdot|x)$ is uniformly bounded away from zero and is Lipchitz in $y$ by the argument in the second step under Assumption \ref{a:primitive} (ii), (iii), and (iv), (\ref{eq:dqdt}) is well defined when it is evaluated at $x=x^+_0$. We can then conclude that it is Lipschitz using the Lipschitzness of $Q(\tau|x_0)$, which also follows from Assumption \ref{a:primitive} (v). The same reasoning applies to $\frac{\partial}{\partial \tau}Q(\tau|x^-_0)$. This shows that Assumption \ref{a:consistency} (iii) (a) holds.
	
	Fourth, Assumption \ref{a:primitive} (v) implies that $Q(\tau|\cdot)$ is continuous at $x_0$ for each $\tau\in T$.
	This shows that the first statement of Assumption \ref{a:consistency} (iii) (b) holds. Finally, 
note that, by Assumption \ref{a:primitive} (i)-(v) and by the chain rule,
\small
\begin{align*}
&\frac{\partial^2 }{\partial x^2}Q(\tau|x)\\
=&\frac{[c_1 f_1(Q(\tau|x),x)-c_2f_2(Q(\tau|x),x)] c_3\frac{\partial}{\partial x}f_3(Q(\tau|x),x)-c_3 f_3(Q(\tau|x),x)\frac{\partial}{\partial x}[c_1 f_1(Q(\tau|x),x)-c_2f_2(Q(\tau|x),x)]}{c_3^2 f^2_3(Q(\tau|x),x)}
\end{align*}
\normalsize
exists and is Lipschitz. A similar argument holds for higher order derivatives.
This shows that the second statement of Assumption \ref{a:consistency} (iii) (b) holds. 
\end{proof}



\section{Practical Guideline}\label{sec:guide_to_practice}

\subsection{Bandwidth Choice}\label{sec:guide_to_practice_bandwidth}
This section presents a guide to practice for bandwidth choice. 
Imbens and Kalyanaraman (2012), Calonico, Cattaneo and Titiunik (2014), and Arai and Ichimura (2016) provide data-driven optimal bandwidth selection algorithms for the mean regression discontinuity design.
In this section, we propose a bandwidth selection rule based on the MSE for the local linear estimation of the conditional CDF, which is compatible with orders $p>1$ for biased-corrected estimation.

We define the following notations: $ u_1=[1, ud^+_u, ud^-_u]'$, $ N_1=\int_{\mathds{R}} [1,ud^+_u, ud^-_u] [1,ud^+_u, ud^-_u]'_1 K(u)du$, $ T_1=(c(\tau))^{-1}\int_{\mathds{R}}  [1,ud^+_u, ud^-_u]  [1,ud^+_u, ud^-_u]' K^2(u)du$ and $ \jmath_2=[0,1,0]'$, $ \jmath_3=[0,0,1]'$.
With the order of polynomial set to one, Lemma \ref{lemma:bahadur} and Theorem \ref{theorem:weakconv1} together imply that the approximate MSE is
$
MSE(\hat \beta^+_1(\tau)-\hat\beta^-_1(\tau))
=Bias^2( \hat \beta^+_1(\tau)-\hat\beta^-_1(\tau) ) + Var( \hat \beta^+_1(\tau)-\hat\beta^-_1(\tau) ),
$
where
\begin{align*}
Bias(\hat \beta^+_1(\tau)-\hat\beta^-_1(\tau))=&h_{n,\tau}\Big[\frac{ \jmath'_2 ( N_1)^{-1}}{2!} \int_{\mathds{R}} \Big( \frac{\partial^{2} Q(\tau|x^+_0)}{\partial x^{2}}d^+_u + \frac{\partial^{2} Q(\tau|x^-_0)}{\partial x^{2}}d^-_u \Big)  u^{2}    u'_1 K(u) du\\
&-\frac{\jmath'_3 (N_1)^{-1}}{2!} \int_{\mathds{R}} \Big( \frac{\partial^{2} Q(\tau|x^+_0)}{\partial x^{2}}d^+_u + \frac{\partial^{2} Q(\tau|x^-_0)}{\partial x^{2}}d^-_u \Big)  u^{2}     u'_1 K(u) du\Big] \qquad\text{and}\\
Var(\hat \beta^+_1(\tau)-\hat\beta^-_1(\tau))=&\frac{1}{nh^3_{n,\tau}}\frac{ \tau(1-\tau)( \jmath'_2 N^{-1}_1 T_1(\tau)  N^{-1}_1 \jmath_2+ \jmath'_3 N^{-1}_1 T_1(\tau)  N^{-1}_1 \jmath_3-2\jmath'_2 N^{-1}_1 T_1(\tau)  N^{-1}_1 \jmath_3)}{f_X(x_0)f_{Y|X}(Q(\tau|x_0)|x_0)}
\end{align*}
Taking the first order condition with respect to the bandwidth, 
	under Assumption \ref{a:consistency}, we obtain the approximate optimal choice of $h_{n,\tau}$ for the QRKD estimand:
	\begin{align*}
	h^*_{n,\tau}(s)=\Big( \frac{3}{2}\frac{C_2(\tau)}{C^2_1(\tau)} \Big)^{\frac{1}{2}}n^{-\frac{1}{5}},
	\end{align*}
	where
	\begin{align*}
	C_1(\tau)=&\frac{ \jmath'_2 ( N_1)^{-1}}{2!} \int_{\mathds{R}} \Big( \frac{\partial^{2} Q(\tau|x^+_0)}{\partial x^{2}}d^+_u + \frac{\partial^{2} Q(\tau|x^-_0)}{\partial x^{2}}d^-_u \Big)  u^{2}    u'_1 K(u) du\\
	-&\frac{\jmath'_3 (N_1)^{-1}}{2!} \int_{\mathds{R}} \Big( \frac{\partial^{2} Q(\tau|x^+_0)}{\partial x^{2}}d^+_u + \frac{\partial^{2} Q(\tau|x^-_0)}{\partial x^{2}}d^-_u \Big)  u^{2}     u'_1 K(u) du \qquad\text{and}\\
	C_2(\tau)=&\frac{ \tau(1-\tau)( \jmath'_2 N^{-1}_1 T_1(\tau)  N^{-1}_1 \jmath_2+ \jmath'_3 N^{-1}_1 T_1(\tau)  N^{-1}_1 \jmath_3-2\jmath'_2 N^{-1}_1 T_1(\tau)  N^{-1}_1 \jmath_3)}{f_X(x_0)f_{Y|X}(Q(\tau|x_0)|x_0)}.
	\end{align*}

For bias-corrected estimation with an order $p>1$, this bandwidth rule above provides a rate that is required by Assumption \ref{a:consistency} (v).
In the above formulas, the unknown densities, $f_X$ and $f_{Y \mid X}$, and the unknown conditional quantile function $Q$ and its derivative $\frac{\partial^2}{\partial x^2}Q$ need to be replaced by the respective estimates $\hat f_X$, $\hat  f_{Y \mid X}$, $\check \alpha$, and $\check \beta^\pm_2$.
We thus propose to replace $C_1(\tau)$ and $C_2(\tau)$ by
\begin{align*}
\widehat C_1(\tau)=&\frac{ \jmath'_2 ( N_1)^{-1}}{2!} \int_{\mathds{R}} \Big( \check \beta^+_2(\tau) d^+_u + \check \beta^-_2(\tau)d^-_u \Big)  u^{2}    u'_1 K(u) du\\
-&\frac{\jmath'_3 (N_1)^{-1}}{2!} \int_{\mathds{R}} \Big( \check \beta^+_2(\tau) d^+_u + \check \beta^-_2(\tau)d^-_u \Big)  u^{2}     u'_1 K(u) du \qquad\text{and}\\
\widehat C_2(\tau)=&\frac{ \tau(1-\tau)( \jmath'_2 N^{-1}_1 T_1(\tau)  N^{-1}_1 \jmath_2+ \jmath'_3 N^{-1}_1 T_1(\tau)  N^{-1}_1 \jmath_3-2\jmath'_2 N^{-1}_1 T_1(\tau)  N^{-1}_1 \jmath_3)}{\hat f_X(x_0)\hat f_{Y|X}(\check \alpha(\tau) |x_0)},
\end{align*}
respectively, where
\begin{align*}
\check \alpha (\tau)=&\iota'_1\underset{(\alpha, \beta^+_1, \beta^-_1,\beta^+_2, \beta^-_2)\in \mathds{R}^{5}}{\text{argmin}} \sum_{i=1}^{n}\rho_\tau \Big( y_i- \alpha - \sum_{v=1}^{2}(\beta^+_v d^+_i + \beta^-_v d^-_i)\frac{(x_i-x_0)^v}{v!} \Big),\\
\check \beta^+_2(\tau)=&\iota'_4\underset{(\alpha, \beta^+_1, \beta^-_1,\beta^+_2, \beta^-_2)\in \mathds{R}^{5}}{\text{argmin}} \sum_{i=1}^{n}\rho_\tau \Big( y_i- \alpha - \sum_{v=1}^{2}(\beta^+_v d^+_i + \beta^-_v d^-_i)\frac{(x_i-x_0)^v}{v!} \Big), \qquad\text{and}\\
\check \beta^-_2(\tau)=&\iota'_5\underset{(\alpha, \beta^+_1, \beta^-_1,\beta^+_2, \beta^-_2)\in \mathds{R}^{5}}{\text{argmin}} \sum_{i=1}^{n}\rho_\tau \Big( y_i- \alpha - \sum_{v=1}^{2}(\beta^+_v d^+_i + \beta^-_v d^-_i)\frac{(x_i-x_0)^v}{v!} \Big).
\end{align*}

Bandwidth choices for the preliminary estimates, $\hat f_X$ and $\hat f_{Y \mid X}$, can be conducted by standard rule-of-thumb or data-driven methods.
Let $h^x_{n}$ and $(\bar h^y_{n}, \bar h^x_{n})^\prime$ denote the bandwidths used for estimating $\hat f_X$ and $\hat f_{Y \mid X}$, respectively.
First, $h^x_{n}$ may be obtained by minimizing approximate MISE.
In other words,
$
h^x_n=\big( \int u^2 K(u)du \big)^{-2/5} \big( \int K(u)^2 du \big)^{1/5}\big( \frac{3}{8\sqrt{\pi} } \sigma^{-5}_X \big)^{-1/5} n^{-1/5}
$, where $\sigma_X$ can be estimated by sample variance of $X$.
See Sections 3.3 and 3.4 of Silverman (1986).
Second, Bashtannyk and Hyndman (2001) suggest that $(\bar h^y_{n}, \bar h^x_{n})^\prime$ may be obtained by
\begin{align*}
(\bar h^y_{n}, \bar h^x_{n})^\prime=\bigg( \Big( \frac{d^2v}{2.85\sqrt{2\pi} \sigma^5_{X} }\Big)^{1/4}\bar h^x_{n}, \Big( \frac{32 R^2(K)\sigma^5_{Y} (260 \pi^9 \sigma^{58}_{X} )^{1/8}}{n \sigma^4_K d^{5/2} v^{3/4} [v^{1/2} + d(16.25 \pi \sigma^{10}_{X})^{1/4}]} \Big)^{1/6},
\bigg)^\prime,
\end{align*}
where $R(K)=\int K^2(u)du$, $v= 0.95  \sqrt{2\pi} \sigma^3_{X} (3 d^2 \sigma^2_{X} + 8\sigma^2_{Y})-32\sigma^2_{X} \sigma^2_{Y} \exp(-2)$, and $d$ is the slope of an OLS of $y_i$ on $[1,x_i]'$. 
Here, $\sigma^2_K$ is the variance with respect to the kernel function $K$.
The variances, $\sigma^2_{X}$ and $\sigma^2_{Y}$, can be replaced by sample variances of $x_i$ and $y_i$, respectively.


\subsection{Pivotal Simulation and Implementation of Uniform Inference }\label{sec:bias}

As pointed out in Section 6 of Qu and Yoon (2015a) and Remark 2 of Qu and Yoon (2015b), the distribution of the process $G(\tau,j)$ is conditionally pivotal and the randomness of Uniform Bahadur Representations come only from $\{\tau-\mathds{1}\{y_i\le Q(\tau|x_i)\}\}^n_{i=1}$ conditional on data.
In this light, we can simulate the distribution of $G(\tau,j)$ in the following manner.
In each iteration, we generate $\{u_i\}^n_{i=1} \stackrel{i.i.d.}{\sim} Uniform(0,1)$ independently from data, and evaluate $\{\tau-\mathds{1}\{u_i\le \tau\}\}^n_{i=1}$ in place of $\{\tau-\mathds{1}\{y_i\le Q(\tau|x_i)\}\}^n_{i=1}$ in the Uniform Bahadur representations.
Repeat this process many times.
With this procedure, we can perform the tests of significance and heterogeneity as in Section \ref{sec:test} via simulating the supremum of $G(\tau,j)$.
The following algorithm presents a complete procedure to implement the non-standardized test of significance and test of heterogeneity in corollary \ref{corollary:test}.

\begin{algorithm}\label{algorithm:SQRKD_non-std}\ \
	\begin{enumerate}
		\item Discretize $T$ into a grid points $T_d=\{t_1,...,t_T\}$. For each $\tau\in T_d$, estimate $(\hat \alpha(\tau),\hat\beta^+_1(\tau),\hat\beta^-_1(\tau))$. 
		\item Estimate $\hat f_X(x_0)$ and estimate $\hat f_{Y|X}(\hat \alpha(\tau)|x_0)$ for each $\tau\in T_d$.
		\item Generate $\{u_i\}^n_{i=1} \stackrel{i.i.d.}{\sim} Uniform(0,1)$ independently from data.
		\item For each $\tau\in T_d$, compute
	\begin{align*}
\widehat Y_1(\tau)=\frac{(\iota'_2 -\iota'_3)N^{-1} \sum_{i=1}^{n} z_{i,n,\tau} K_{i,n,\tau} (\tau-\mathds{1}\{u_i\le \tau \}) }{(b'(x^+_0)-b'(x^-_0))\sqrt{nh_{n,\tau}}\hat f_X(x_0) \hat f_{Y|X}(\hat \alpha(\tau)|x_0)}
	\end{align*}
	\item Iterate the third and fourth steps $M$ times to obtain $\{\widehat Y_j(\cdot)\}_{j=1}^{M}$ on $T_d$.
	\item Compute the test statistic(s) $WS_n(T_d)$ and/or $WH_n(T_d)$.
	\item Compute the $p$-th quantile(s) of $\max_{\tau \in T_d}|\widehat Y_j(\tau)|$ and/or $\max_{\tau \in T_d}|\phi'_{QRKD}(\widehat Y_j)(\tau)|$, the simulated critical values for the test statistic(s), $WS_n(T_d)$ and/or $WH_n(T_d)$, respectively.
	\end{enumerate}
\end{algorithm}

To compute standardized version of the test statistics, we also need to obtain estimates for $\widehat \sigma^s$ and $\widehat \sigma^h$. 
We compute them based on the standard deviations of
\begin{align*}
A_n^s(\tau)&=\sqrt{n h_{n,\tau}^3}\widehat{QRKD}(\tau)
\qquad\text{and}\\
A_n^h(\tau)&=\sup_{\tau \in T}\sqrt{n h_{n,\tau}^3}\Big[\widehat{QRKD}(\tau)- |T|^{-1}\int_{T} \widehat{QRKD}(\tau') d\tau^\prime\Big].
\end{align*} 
The following algorithm outlines a procedure for the standardized version of the test.

\begin{algorithm}\label{algorithm:SQRKD_std}
Steps 1--5 remain the same as those in Algorithm \ref{algorithm:SQRKD_non-std}.
\begin{enumerate}
	\setcounter{enumi}{5}
	\item Compute the test statistic(s) $WS^{std}_n(T_d)$ and/or $WH^{std}_n(T_d)$.
	\item Compute the $p$-th quantile(s) of $\max_{\tau \in T_d}|\widehat Y_j(\tau)/\widehat \sigma^s(\tau)|$ and/or $\max_{\tau \in T_d}|\phi'_{QRKD}(\widehat Y_j)(\tau)/\widehat \sigma^h(\tau)|$, the simulated critical values for the test statistic(s), $WS^{std}_n(T_d)$ and/or $WH^{std}_n(T_d)$, respectively.
\end{enumerate}
\end{algorithm}


\begin{thebibliography}{9}




\bibitem[Angrist, Graddy and Imbens(2000)]{angrist_graddy_imbens2000}
Angrist, Joshua D., Kathryn Graddy, and Guido W. Imbens (2000)
``The Interpretation of Instrumental Variables Estimators in Simultaneous Equations Models with an Application to the Demand for Fish,''
\textit{Review of Economic Studies,} Vol. 67, No. 3, pp. 499--527.

\bibitem[Angrist and Imbens (1995)]{angrist_imbens1995}
Angrist, Joshua D. and Guido W. Imbens (1995)
``Two-Stage Least Squares Estimation of Average Causal Effects in Models with Variable Treatment Intensity,''
\textit{Journal of the American Statistical Association,} Vol. 90, No. 430, pp. 431--442.

\bibitem[Arai and Ichimura(2016) ]{arai_ichimura2016}
Arai, Yoichi, and Hidehiko Ichimura (2016) 
``Optimal Bandwidth Selection for the Fuzzy Regression Discontinuity Estimator,'' 
\textit{Economics Letters,} Vol 141, pp. 103--106.


\bibitem[Bashtannyk and Hyndman (2001)]{bashtannyk_hyndman2001}
Bashtannyk, David M., and Rob J. Hyndman (2001)
``Bandwidth selection for kernel conditional density estimation,''
\textit{Computational Statistics and Data Analysis,}
Vol. 36, No. 3,  pp. 279--298.

\bibitem[Card, Lee, Pei and Weber (2016)]{card_lee_pei_weber2016}
Card, David, David Lee, Zhuan Pei, and Andrea Weber (2016)
 ``Inference on Causal Effects in a Generalized Regression Kink Design,''
\textit{Econometrica,} Vol. 83, No. 6, pp. 2453--2483.

\bibitem[Calonico, Cattaneo and Titiunik (2014)]{calonico_cattaneo_titiunik2014}
Calonico, Sebastian, Matias D. Cattaneo, and Rocio Titiunik (2014)
 ``Robust Nonparametric Confidence Intervals for Regression-Discontinuity Designs,''
\textit{Econometrica,} Vol. 82, No. 6, pp. 2295--2326.

\bibitem[Calonico, Cattaneo, Farrell and Titiunik (2016)]{calonico_cattaneo_farrell_titiunik2016}
Calonico, Sebastian, Matias D. Cattaneo, Max Farrell and Rocio Titiunik (2016)
``Regression Discontinuity Designs Using Covariates,''
\textit{Working Paper}.

\bibitem[Cattaneo and Escanciano(2016)]{cattaneo_escanciano2016}
Cattaneo, Matias D., and Juan Carlos Escanciano (2016) ``Regression Discontinuity Designs: Theory and Applications,''
\textit{Advances in Econometrics,} Vol. 38 (Forthcoming)


\bibitem[Chernozhukov and Fern\'andez-Val (2005)]{chernozhukov_fernandezval2005}
Chernozhukov, Victor and Iv\'an Fern\'andez-Val (2005)
``Subsampling Inference on Quantile Regression Processes,''
\textit{Sankhya: The Indian Journal of Statistics,} Vol. 67, No. 2, pp. 253--276.



\bibitem[Cook (2008)]{cook2008}
Cook, Thomas D. (2008) ``Waiting for Life to Arrive: a History of the Regression-Discontinuity Design in Psychology, Statistics and Economics,'' \textit{Journal of Econometrics,} Vol. 142, No. 2, pp. 636--654.

\bibitem[Dong (2016)]{dong2016}
Dong, Yingying (2016) ``Jump or Kink? Identifying Education Effects by Regression Discontinuity Design without the Discontinuity,'' Working Paper.




\bibitem[Frandsen, Fr\"olich and Melly(2012)]{frandsen_frolich_melly2012}
Frandsen, Brigham R., Markus Fr\"olich and Blaise Melly (2012)
``Quantile Treatment Effects in the Regression Discontinuity Design,''
\textit{Journal of Econometrics,} Vol. 168, No.2  pp. 382-395.


\bibitem[Guerre and Sabbah (2012)]{guerre_sabbah2012}
Guerre, Emmanuel and Camille Sabbah (2012)
``Uniform Bias Study and Bahadur Representation for Local Polynomial Estimators of the Conditional Quantile Function,''
\textit{Econometric Theory,} Vol. 26, No. 5, pp. 1529-1564.


\bibitem[Heckman and Vytlacil(2005)]{heckman_vytlacil2005}
Heckman, James J. and Edward Vytlacil (2005)
``Structural Equations, Treatment Effects, and Econometric Policy Evaluation,''
\textit{Econometrica,} Vol. 73, No. 3, pp. 669--738.



\bibitem[Imbens and Lemieux (2008)]{imbens_lemieux2008}
Imbens, Guido and Thomas Lemieux (2008)
``Special Issue Editors' Introduction: The Regression Discontinuity Design  -- Theory and Applications,''
\textit{Journal of Econometrics,} Vol. 142, No. 2, pp. 611--614.

\bibitem[Imbens and Wooldridge (2009)]{imbens_wooldridge2009}
Imbens, Guido W. and Jeffrey M. Wooldridge (2009)
``Recent Developments in the Econometrics of Program Evaluation,''
\textit{Journal of Economic Literature,} Vol. 47, No. 1, pp. 5--86.

\bibitem[Imbens and Kalyanaraman (2012)]{imbens_kalyanaraman2011}
Imbens, Guido W. and Karthik Kalyanaraman (2012)
``Optimal Bandwidth Choice for the Regression Discontinuity Estimator,''
\textit{Review of Economic Studies,} Vol. 79, No. 3, pp. 933--959.

\bibitem[Kato and Sasaki(2017)]{kato_sasaki2017}
Kato, Ryutah and Yuya Sasaki (2017)
``On Using Linear Quantile Regressions for Causal Inference.''
\textit{Econometric Theory,} Vol. 33, No. 3, pp. 664--690.

\bibitem[Koenker (2005)]{koenker2005}
Koenker, Roger (2005)
``Quantile Regression,''
Cambridge University Press: Cambridge.


\bibitem[Koenker and Xiao (2002)]{koenker_xiao2002}
Koenker, Roger and Zhijie Xiao (2002)
``Inference on the quantile regression process,''
\textit{Econometrica,} Vol. 70, No. 4, pp.1583--1612.


\bibitem[Kong, Linton and Xia (2010)]{kong_linton_xia2010}
Kong, Efang, Oliver B. Linton, and Yingcun Xia (2010)
``Uniform Bahadur Representation for Local Polynomial Estimates of M-Regression and its Application to the Additive Model,''
\textit{Econometric Theory,} Vol. 26, No. 5, pp. 1529-1564.


\bibitem[Landais (2011)]{landais2011}
Landais, Camille (2011)
``Heterogeneity and Behavioral Responses to Unemployment Benefits over the Business Cycle,''
Working Paper, LSE.

\bibitem[Landais (2015)]{landais2015}
Landais, Camille (2015)
``Assessing the Welfare Effects of Unemployment Benefits Using the Regression Kink Design,''
\textit{American Economic Journal: Economic Policy}, Vol. 7, No. 4, pp. 243--278.

\bibitem[Lee and Lemieux (2010)]{lee_lemieux2010}
Lee, David S., and Thomas Lemieux (2010)
``Regression Discontinuity Designs in Economics,''
\textit{Journal of Economic Literature,} Vol. 48, No. 2, pp. 281--355.

\bibitem[Moffitt (1985)]{moffitt1985}
Moffitt, Robert (1985)
``The Effect of the Duration of Unemployment Benefits on Work Incentives: An Analysis
of Four Datasets,''
Unemployment Insurance Occasional Papers 85-4, U.S. Department of Labor, Employment
and Training Administration.

\bibitem[Nielsen, S\o rensen, and Taber (2010)]{nielsen_sorensen_taber2010}
Nielsen, Helena Skyt, Torben S\o rensen, and Christopher Taber (2010)
``Estimating the Effect of Student Aid on College Enrollment: Evidence from a Government Grant Policy Reform,''
\textit{American Economic Journal: Economic Policy,} Vol. 2, No. 2, pp. 185--215.


\bibitem[Padula (2011)]{padula2011}
Padula, Mariarosaria. (2011) Asymptotic Stability of Steady Compressive Fluids. Springer.


\bibitem[Qu and Yoon(2015a)]{qu_yoon2015a}
Qu, Zhongjun and Jungmo Yoon (2015a)
``Nonparametric Estimation and Inference on Conditional Quantile Processes,''
\textit{Journal of Econometrics,} Vol. 185, No.1  pp. 1-19.

\bibitem[Qu and Yoon(2015b)]{qu_yoon2015b}
Qu, Zhongjun and Jungmo Yoon (2015b)
``Uniform Inference on Quantile Effects under Sharp Regression Discontinuity Designs,''
\textit{Working Paper,} 2015.

\bibitem[Sabbah (2014)]{sabbah2014}
Sabbah, Camille (2014)
``Uniform Confidence Bands for Local Polynomial Quantile Estimators,''
\textit{ESAIM: Probability and Statistics,} Vol. 18, pp. 265-276.

\bibitem[Sasaki (2015)]{sasaki2015}
Sasaki, Yuya (2015)
``What Do Quantile Regressions Identify for General Structural Functions?,''
\textit{Econometric Theory,} Vol. 31, No. 5, pp. 1102-1116.

\bibitem[Silverman(1986)]{silverman1986}
Silverman, Bernard W. (1986)
``Density Estimation for Statistics and Data Analysis,''
Chapman \& Hall/CRC: London.

\bibitem[Simonsen, Skipper and Skipper (2015)]{simonsen_skipper_skipper2015}
Simonsen, Marianne, Lars Skipper, and Niels Skipper (2015)
``Price sensitivity of demand for prescription drugs: Exploiting a regression kink design,''
\textit{Journal of Applied Econometrics,} Forthcoming.

\bibitem[van der Vaart (1998)]{vandervaart1998}
van der Vaart, Aad W. (1998)
``Asymptotic Statistics,''
Cambridge University Press: Cambridge.



\bibitem[Yitzhaki(1996)]{yitzhaki1996}
Yitzhaki, Shlomo (1996) ``On Using Linear Regressions in Welfare Economics,''
\textit{Journal of Business and Economic Statistics,} Vol. 14, No. 4, 478--486.

\end{thebibliography}
\end{document}